\documentclass[aps,pra,twocolumn,amsmath,amssymb,nofootinbib,showpacs,superscriptaddress]{revtex4-1}
\usepackage[english]{babel}
\usepackage{latexsym}
\usepackage{graphics}
\usepackage{graphicx}
\usepackage{epsfig}
\usepackage{color}
\usepackage{bm}
\usepackage{amsmath}
\usepackage{amssymb}
\usepackage{amsthm}
\usepackage{dcolumn}
\usepackage{bm}
\usepackage{float}
\usepackage{hyperref}
\usepackage{color}
\usepackage{epstopdf}
\usepackage{cleveref}
\usepackage[svgnames]{xcolor}
\hypersetup{hidelinks,colorlinks=true,allcolors=DarkBlue}

\DeclareMathOperator*{\argmin}{arg\,min}

\begin{document}

\newcommand{\ket}[1]{|#1\rangle}
\newcommand{\bra}[1]{\langle #1|}
\newcommand{\bS}{\text{$\mathbf{S}$}}
\newcommand{\braket}[2]{\langle #1|#2\rangle}

\newcommand{\inp}{{\rm in}}
\newcommand{\out}{{\rm out}}

\newcommand{\tr}{{\rm Tr}}
\newcommand{\init}{{\textrm{init}}}

\newcommand{\kett}[1]{|#1\rangle\rangle}
\newcommand{\braa}[1]{\langle\langle #1|}
\newcommand{\com}[1]{{\color{blue}#1}}

\newtheorem{theorem}{Theorem}
\newtheorem{corollary}{Corollary}
\newtheorem{claim}{Claim}
\newtheorem{Example}{Example}

\preprint{APS/123-QED}

\title{Probability representation of quantum dynamics using pseudostochastic maps}
	
	\author{E.O. Kiktenko}
	\affiliation{Russian Quantum Center, Skolkovo, Moscow 143025, Russia}
	\affiliation{Department of Mathematical Methods for Quantum Technologies, Steklov Mathematical Institute of Russian Academy of Sciences, Moscow 119991, Russia}
	\affiliation{Moscow Institute of Physics and Technology, Dolgoprudny 141700, Russia}

	\author{A.O. Malyshev}
	\affiliation{Russian Quantum Center, Skolkovo, Moscow 143025, Russia}
	\affiliation{Moscow Institute of Physics and Technology, Dolgoprudny 141700, Russia}
		
	\author{A.S. Mastiukova}
	\affiliation{Russian Quantum Center, Skolkovo, Moscow 143025, Russia}
	\affiliation{Moscow Institute of Physics and Technology, Dolgoprudny 141700, Russia}
	
	\author{V.I. Man'ko}
	\affiliation{Russian Quantum Center, Skolkovo, Moscow 143025, Russia}
	\affiliation{Moscow Institute of Physics and Technology, Dolgoprudny 141700, Russia}
	\affiliation{P.N. Lebedev Physical Institute, Russian Academy of Sciences, Moscow 119991, Russia}

	\author{A.K. Fedorov}
	\affiliation{Russian Quantum Center, Skolkovo, Moscow 143025, Russia}
	\affiliation{Moscow Institute of Physics and Technology, Dolgoprudny 141700, Russia}

	\author{D. Chru{\'s}ci{\'n}ski}
	\affiliation{Institute of Physics, Faculty of Physics, Astronomy and Informatics, Nicolaus Copernicus University, Toru{\'n} 87100, Poland}
	
	\date{\today}
	\begin{abstract}
		In this work, we consider a probability representation of quantum dynamics for finite-dimensional quantum systems with the use of pseudostochastic maps acting on true probability distributions.
		These probability distributions are obtained via symmetric informationally complete positive operator-valued measure (SIC-POVM) and can be directly accessible in an experiment.
		We provide SIC-POVM probability representations both for unitary evolution of the density matrix governed by the von Neumann equation and dissipative evolution governed by Markovian master equation. 
		In particular, we discuss whereas the quantum dynamics can be simulated via classical random processes in terms of the conditions for the master equation generator in the SIC-POVM probability representation.
		We construct practical measures of nonclassicality non-Markovianity of quantum processes and apply them for studying experimental realization of quantum circuits realized with the IBM cloud quantum processor.
	\end{abstract}
\maketitle
	
\section{Introduction}

Quantum technologies require an efficient toolbox for synthesis, control, and characterization of the quantum states and processes~\cite{NielsenChuang}.
The task of characterizing quantum states has quite a rich history of attempts to describe quantum systems with the use of standard methods of statistical physics,
such as phase-space probability distributions~\cite{Wigner1932,Husimi1940,Glauber1963,Sudarshan1963}.
A quantum analog of classical phase-space probability distributions, known as the Wigner function, cannot be fully interpreted as a probability distribution because it takes negative values in some cases~\cite{Wigner1932,Wigner1984,Ferry2018}.
The negativity of the Wigner quasiprobability distribution plays an important role in the modern quantum theory since this effect is a signature of the highly nonclassical character of a quantum state.
In particular, it was demonstrated that negativity and contextuality are equivalent notions of nonclassicality~\cite{Spekkens2008}.
The negativity of the Wigner function has been largely studied for quantum information processing both for systems
with continuous~\cite{Wigner1984} and discrete variables~\cite{Wootters1987,Wootters2004,Gross2006,Ferrie2009,Ferrie2011,Zhu2016}.
Recent progress in quantum information science has helped us to understand the role of properties of quasiprobability distribution
in the context of verifying quantum resources that provide quantum speedup~\cite{Galvao2006,Cormick2006,Gottesman2012,Gottesman2014,Howard2014,Raussendorf2015,Pashayan2015}.

\begin{figure}
	\includegraphics[width=\linewidth]{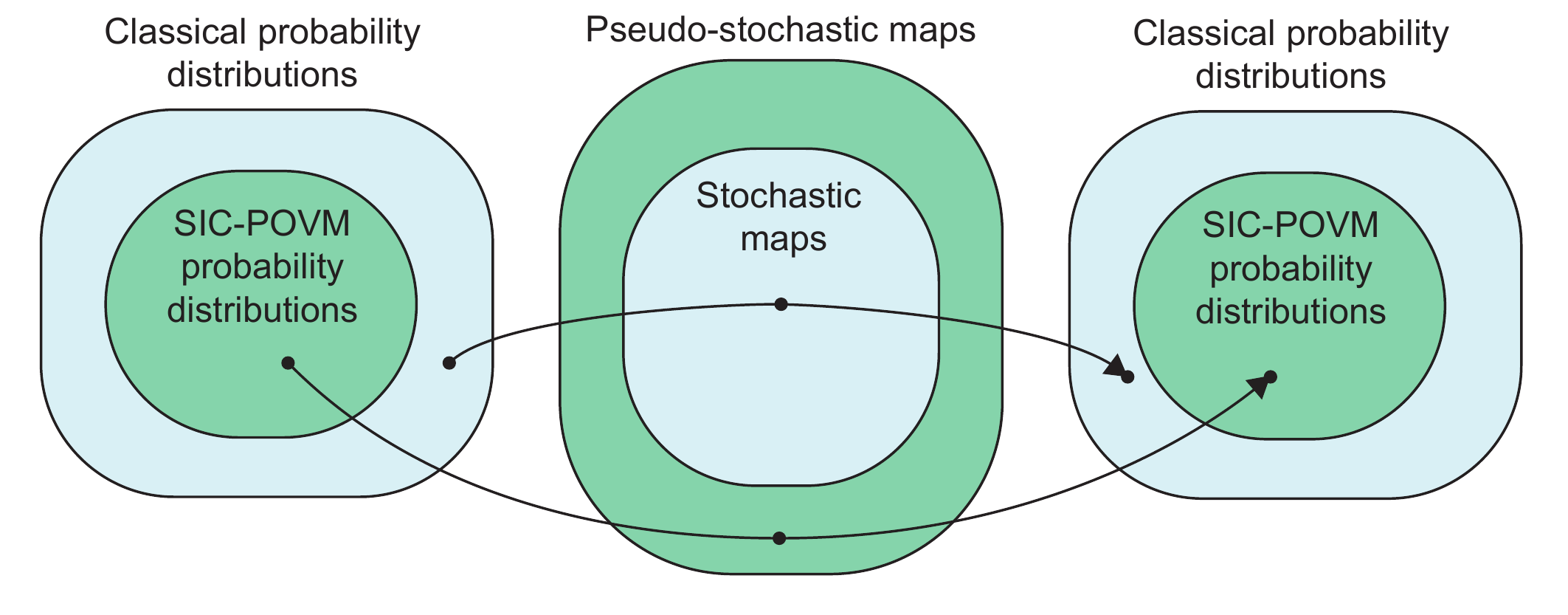}
	\caption{
	General relation among the set of SIC-POVM probability distributions, general classical probability distributions of the same dimension, and maps which keep distributions in a particular set.
	The set of SIC-POVM probability distributions, so-called Hilbert \emph{qplex}~\cite{Appleby2017}, is smaller than the full probability distribution set,
	while the set of possible maps turning SIC-POVM probabilities into SIC-POVM probabilities (known to be a pseudostochastic) is wider than the set of classical stochastic maps.
	Points and arrows demonstrate actions of pseudostochastic (stochastic) maps on SIC-POVM (classical) probability distributions.}
	\label{fig:intuition}
\end{figure}

An approach to the description of quantum phenomena using the language of probability distributions~\cite{Chernega2017,Chernega2018, Avanesov2019,Avanesov2019b} 
has been extended by the concept of informationally complete POVMs (IC-POVMs) and SIC-POVMs
that are based on measurements completely describing quantum states~\cite{Caves2002,Caves2004,Filippov2010}.
In this framework, quantum states are associated with probabilities related to a specific set of vectors in the Hilbert space, which is composed of so-called SIC projectors.
It is important to point out that in the SIC approach, the probability distributions describing quantum states contain no redundant information; i.e., the number of probabilities is the minimum possible for reconstructing all density-matrix elements.
We note that analytic proofs of SIC existence have only been found in a number of cases~\cite{Fuchs2017}.
The approach of describing quantum states with SIC-POVM probability distributions has been widely explored in quantum Bayesianism (QBism) reformulation of quantum mechanics~\cite{Appleby2011,Fuchs2013}.
SIC-POVM measurements also has been actively used in various experiments~\cite{Rehacek2004,Durt2008,Medendorp2011,Bent2015,Zhao2015,Hou2018}.
Importantly, it turns out that the set of possible probability distributions obtained with SIC-POVM measurements is smaller than the full probability distributions set of the same dimension.
This ``quantum'' part of a classical probability simplex, which is achievable via SIC-POVM measurements, is referred to as a {``qplex''}~\cite{Appleby2017}.
An important result of Ref.~\cite{Appleby2017} is the derivation of the properties of qplexes from the very fundamental assumptions about quantum theory as well as
the description of a link between the symmetry properties of qplexes and a condition for the existence of a $d$-dimensional SIC-POVM.

The representation of quantum states with the use of quasiprobability distributions can be further generalized 
to the representation of quantum processes (channels) with the use of quasistochastic matrices~\cite{Chruscinski2013,Chruscinski2015,Zhu2016b,Wetering2017}.
In contrast to traditionally used stochastic matrices, which describe the evolution of classical probability distributions, quasi-stochastic matrices can posses negative elements.
In Ref.~\cite{Wetering2017}, a functorial embedding of the quantum channels category into the category of quasi-stochastic matrices has been provided.
Thus, the formalism of quasi-stochastic matrices can serve as an alternative formulation of the quantum theory and looks promising in the framework of quantum resources analysis.
This apparatus, however, has not been consistently applied to quantum information processing tasks yet.

In this work, we focus on the dynamics of probability distributions obtained in SIC-POVM measurements.
In line with Refs.~\cite{Chruscinski2013,Chruscinski2015}, we refer to the resulting matrices (maps), which define the evolution of SIC-POVM probability distributions as \emph{pseudostochastic} rather than quasistochastic.
This is because these matrices correspond to the evolution of true experimentally accessible probabilities rather then quasi-probabilities.
As was mentioned, pseudostochastic matrices are an analog of conditional probability matrices without restrictions of positivity for matrix elements (see Fig.~\ref{fig:intuition}).
In fact, pseudostochastic matrices appear in the discussion of unitary evolution of quantum states described by qplexes (see Appendix B of Ref.~\cite{Appleby2017}) and are discussed in detail in Ref.~\cite{Wetering2017}.
However, a connection between pseudo(quasi)stochastic matrices and common equations used for describing quantum system dynamics, to the best of our knowledge, is not considered in the literature.
It is then important to investigate how pseudostochastic matrices appear as solutions of quantum dynamical equation, and we address this question in this work.
More precisely, we derive a dynamics equation for a SIC-POVM probability vector, which corresponds to the von Neumann equation and dissipative evolution governed by Markovian master equation.
Next we consider a general scheme of quantum mechanical experiments and demonstrate that the developed SIC-POVM representation shines additional light on nonclassical features of quantum process.
We consider necessary and sufficient conditions for the time-independent Markovian generator to produce classical-like quantum evolution.
These conditions allow us to construct practical measure for nonclassicality of quantum processes.
We show that the considered SIC-POVM probability representation also allows introducing new practical measure of non-Markovianity, which has been actively studied in recent decades~\cite{Wolf2008,Breuer2009,Rivas2010,Luo2012,Breuer2012,Hall2014,Chruscinski2014,Bylicka2014,Chruscinski2015b,Torre2015,Pineda2016,Li2019}.
Finally, we apply our theoretical results to experimental study of processes realized on the super-conducting quantum cloud IBM QX4 quantum processor~\cite{IBM}.

Our work is organized as follows.
In Sec.~\ref{sec:general}, we present a general scheme for the SIC-POVM probability representation of states and measurements.
In Sec.~\ref{sec:dynamics}, we derive an equation for a SIC-POVM probability vector which corresponds to the von Neumann equation and the Markovian master equation.
In Sec.~\ref{sec:noncl}, we study a relation between pseudostochasticity and non-classicality, and construct practical measures of nonclassicality and non-Markovianity of quantum processes.
In Sec.~\ref{sec:experiment}, we apply our theoretical result to experimental study of the IBM Q4 superconducting quantum processor.
We summarize main results and conclude in Sec.~\ref{sec:conclusion}.

\section{Probability representation of states, measurements and linear maps}\label{sec:general}

We start our consideration by introducing SIC-POVM effects, which can be used in the construction of probability representation of states and measurements for finite-dimensional systems.
Consider a $d$-dimensional Hilbert space $\mathcal{H}$ with $d\geq2$.
In what follows, we assume that it is possible to find out a set of $d^2$ normalized states $\{\ket{\psi_i}\}_{i=1}^{d^2}$ belonging to $\mathcal{H}$ such that
\begin{equation} \label{eq:scal_prod}
	|\bra{\psi_i}\psi_j\rangle|^2 = \tr(\Pi_i\Pi_j)=\frac{d\delta_{i,j}+1}{d+1},
\end{equation}
where $\Pi_i:=\ket{\psi_i}\bra{\psi_i}$ and $\delta_{i,j}$ stands for the Kronecker symbol.
We note that the set $\{\Pi_i\}_{i=1}^{d^2}$ forms a basis in the space $\mathcal{L}(\mathcal{H})$ of linear operators acting on $\mathcal{H}$.

The set $\{\Pi_i/d\}_{i=1}^{d^2}$ is called a SIC-POVM.
By its definition, we have
$\frac{1}{d}\Pi_i \geq 0$ and $\sum_i\Pi_i = \mathbf{1}d$,
where  $\mathbf{1}$ denotes identity operator.
At this moment, SIC-POVMs are found for $d = 2$--151, 168, 172, 195, 199, 228, 259, 323, and 844.
The obtained solutions are available online~\cite{QubismSite} (for a review, see also Ref.~\cite{Fuchs2017} and reference therein).
We also note that numerical methods are heavily involved for SIC-POVM search~\cite{Scott2010,Scott2017}.

\subsection{Representation of states}

Let us consider a quantum state given by a unit-trace semipositive Hermitian density operator $\rho\in\mathcal{L}(\mathcal{H})$ ($\rho\geq0$, $\tr\rho=1$).
The probability of obtaining an $i$th outcome corresponding to the effect $\Pi_i/d$ after SIC-POVM measurement is given by
\begin{equation} \label{eq:SICPOVMprobs}
	p_i = \frac{1}{d}\mathrm{Tr}(\rho\Pi_i).
\end{equation}
Let us write these probabilities in the form of vector $p :=
	\begin{bmatrix}
	p_1 & \ldots & p_{d^2}
	\end{bmatrix}^{\rm T}$,
which we further refer to as a SIC-POVM probability vector.

The density matrix $\rho$ can be reconstructed back from the SIC-POVM probability vector in the following way:
\begin{multline}\label{eq:rho}
	\rho = \sum_{i=1}^{d^2} \left[(d+1)p_i-\frac{1}{d}\right]\Pi_i\\=\sum_{i=1}^{d^2} \left[(d+1)\Pi_i-\mathbf{1}\right]p_i = \sum_i K_i p_i,
\end{multline}
where $K_i= (d+1)\Pi_i-\mathbf{1}$.

We mention the following useful property of $K_i$:
\begin{equation} \label{eq:property_of_K}
	\tr(K_i\Pi_i)=(d+1)\tr(\Pi_i\Pi_j)-\tr(\Pi_j)=d\delta_{i,j}.
\end{equation}

It is also useful to introduce a vectorized representation of linear operators.
Let $\{\ket{i}\}$ be an orthonormal basis in $\mathcal{H}$ and $A\in\mathcal{L}(\mathcal{H})$ is some linear operator.
We can write
$A=\sum_{i,j}A_{i,j}\ket{i}\bra{j}$,
where $A_{i,j}=\bra{i}A\ket{j}$ are matrix elements of $A$ in the $\{\ket{i}\}$ representation.
Next, we refer to
$\kett{A} := \sum_{i,j}A_{i,j}\ket{i}\otimes\ket{j}\in\mathcal{H}\otimes\mathcal{H}$,
as the ``ket'' vector representation of $A$.
In a similar way, we can introduce an adjoint vector $\braa{A} := \sum_{i,j} A_{i,j}^*\bra{i}\otimes\bra{j}\in\mathcal{H}^*\otimes\mathcal{H}^*$.
It is easy to check that for any two operators $A,B\in\mathcal{L}(\mathcal{{H}})$, their Hilbert-Schmidt product takes the form
$\tr(A^\dagger B)=\braa{A}\kett{B}$.

Let $A,B,U,V\in\mathcal{L}(\mathcal{H})$.
It is also easy to check that the identity $B=UAV^\dagger$ corresponds to the following identities in the vector representation:
\begin{equation} \label{eq:product}
	\kett{B}=U\otimes V^*\kett{A}, \quad \braa{B}=\braa{A} U^\dagger\otimes V^{\rm T}.
\end{equation}

Using the introduced vectorized representation, one can rewrite Eq.~\eqref{eq:rho} in the following form:
\begin{equation} \label{eq:statetoprob}
	\kett{\rho}=\sum_{i=1}^{d^2}\kett{K_i}p_i={\bf K}p,
\end{equation}
where
\begin{multline} \label{eq:Kdef}
	{\bf K} = \begin{bmatrix}
	\kett{K_1} & \ldots & \ket{K_d^2}
	\end{bmatrix}
	\\=
	(d+1)\begin{bmatrix}
	\kett{\Pi_1} & \ldots & \kett{\Pi_d^2}
	\end{bmatrix}-
	\begin{bmatrix}
	\kett{\bf 1} & \ldots & \kett{\bf 1}
	\end{bmatrix}
\end{multline}
is a $d^2\times d^2$ matrix.

Thus, the linear transformation ${\bf K}$ defines the map from the SIC-POVM probability vector $p$ to the vectorized representation of the density matrix $\rho$.
Using Eq.~\eqref{eq:SICPOVMprobs}, we obtain the inverse matrix ${\bf K}^{-1}$ given by
\begin{equation}
	{\bf K}^{-1} = \frac{1}{d}
	\begin{bmatrix}
	\braa{\Pi_1} & \ldots & \braa{\Pi_{d^2}}
	\end{bmatrix}^{\rm T}.
\end{equation}

We also note the following correspondence between the Hilbert-Schmidt product of states and dot product of SIC-POVM probability vectors.
Let $\rho$ and $\sigma$ be two arbitrary density matrices, and $p$ and $s$ be their corresponding SIC-POVM probability vectors.
One can check that
\begin{equation}
	\tr(\rho\sigma)=\braa{\rho}\kett{\sigma}=d(d+1)\langle p,s\rangle-1,
\end{equation}
where $\langle p,s\rangle = \sum_{i=1}^{d^2}p_is_i$ is a standard dot product of two SIC-POVM probability vectors.
Since $\tr(\rho\sigma)\in[0,1]$, we also have
\begin{equation}
	\frac{1}{d(d+1)} \leq \langle p,s\rangle \leq \frac{2}{d(d+1)},
\end{equation}
where the minimum is achieved for two orthogonal states $\rho$ and $\sigma$, while the maximum is achieved for $s=p$ and a pure state $\rho=\sigma=\ket{\psi}\bra{\psi}$.

As was already mentioned in the introduction, the set of possible SIC-POVM probability vectors is smaller than the full set of all possible $d^2$-dimensional probability vectors.
This fact can be easily verified by noticing that the maximum probability in the SIC-POVM probability vector can not exceed the value $d^{-1}$ due to the structure of SIC-POVM effects.
We refer readers to Ref.~\cite{Appleby2017}, where properties of the SIC-POVM probability vectors set (qplex) are studied in detail.

In Appendix~\ref{app:MUB}, we provide a relation between the considered SIC-POVM probability representation and an alternative probability representation based on mutually unbiased measurements (MUB) in the case of $d=2$.

\subsection{Representation of measurements}

Let us now consider a question how the SIC-POVM probability vectors determine the probabilities for arbitrary measurements.
Consider a POVM $E=\{E_1,\ldots,E_m\}$ with $E_i\geq 0$ and $\sum_i E_i={\bf 1}$.
Let us introduce an $m$-dimensional vector
$q:=\begin{bmatrix} q_1 & \ldots & q_m \end{bmatrix}^{\rm T}$
with elements given by probabilities of obtaining different outcomes in measuring $E$ for some state $\rho$: $q_i=\tr(E_i\rho)$.

The relation between the SIC-POVM probability vector $p$ (corresponding to $\rho$) and a new probability vector $q$ for an arbitrary measurement can be obtained using Eq.~\eqref{eq:rho}:
\begin{equation} \label{eq:arbitrmeas}
	q = {\bf M} p,
\end{equation}
where
\begin{equation} \label{eq:Mdef}
	{\bf M}= (d+1){\bf m}-
	\begin{bmatrix}
	\tr E_1 \\ \vdots \\ \tr E_m
\end{bmatrix}
\underbrace{
	\begin{pmatrix}
	1 & \ldots & 1
	\end{pmatrix}}_{d^2~{\rm elements}}
\end{equation}
is $m \times d^2$ matrix with elements given by ${\bf m}_{i,j}=\tr(E_i\Pi_j)$.
The matrix ${\bf m}$ is a stochastic rectangular matrix, that is, it satisfies the following properties:
	${\bf m}_{i,j} \geq 0$, $\sum_{i}{\bf m}_{i,j}=\tr\left(\sum_iE_i\Pi_j\right)=1$.
We note that ${\bf m}$ appears to be bistochastic; that is, its rows also sum up to 1, in the case where $m=d^2$ and $\tr E_i=d^{-1}$
for any $i\in\{1,\ldots, d^2\}$.
It may be the case when $\{E_i\}$ is also a SIC-POVM.

It easy to check from Eq.~\eqref{eq:Mdef} that the matrix ${\bf M}$ is pseudostochastic, that is, the sum of its elements in each column equals to unity: $\sum_{i}{\bf M}_{i,j}=(d+1)-{\rm Tr {\bf 1}}=1$;
however, some elements ${\bf M}_{i,j}$ may be negative.
It may be the case when $E_i$ is proportional to the projector on the state orthogonal to some $\ket{\psi_j}$.
Then we obtain ${\bf M}_{i,j}=-\tr E_j<0$.

\subsection{Representation of linear maps as pseudostochastic matrices}

Here we consider a representation of positive trace-preserving (PTP) and completely positive trace-preserving (CPTP) maps acting on quantum states in the SIC-POVM framework.

First, consider a PTP linear map $\Phi: \mathcal{L}(\mathcal{H}_\inp) \rightarrow \mathcal{L}(\mathcal{H}_\out)$,
where $\mathcal{H}_\inp$ and $\mathcal{H}_\out$ are $d_\inp$- and $d_\out$-dimensional Hilbert spaces, respectively.

Being a PTP map $\Phi$ transforms density operators in $\mathcal{H}_\inp$ into density operators in $\mathcal{H}_\out$, 
so we can consider an action of $\Phi$ on some input state $\rho^{\rm in}$ resulting in output state $\rho^\out=\Phi[\rho^\inp]$.
Let $\{\Pi_i^{\rm in}\}_{i=1}^{d_{\rm in}^2}$ and $\{\Pi_i^{\rm out}\}_{i=1}^{d_{\rm out}^2}$ be SIC-POVM projectors in $\mathcal{H}_\inp$ and $\mathcal{H}_\out$, respectively.
Let $p^\inp$ and $p^\out$ be SIC-POVM probability vectors corresponding to $\rho^\inp$ and $\rho^\out$. 
Simple algebra leads to 
$p^\out = {\bf S}p^\inp$,
where
\begin{eqnarray}
	&{\bf S}_{i,j} &= (d_\inp +1){\bf s}_{i,j} - \frac{1}{d_\out} {\rm Tr}\left[ \Pi_i^{\rm out} \Phi(\mathbf{1}_\inp) \right] , \label{eq:S}\\
	&{\bf s}_{i,j} &= \frac{1}{d_\out}{\rm Tr}\left[ \Pi^\out_i \Phi(\Pi^\inp_j) \right],\label{eq:s}
\end{eqnarray}
and $\mathbf{1}_\inp$ is the identity operator acting in $\mathcal{H}_\inp$.
One can see that ${\bf s}$ is a stochastic matrix: $\sum_i {\bf s}_{i,j}=1$, ${\bf s}_{i,j}\geq 0$, while ${\bf S}_{i,j}$ is pseudostochastic one: $\sum_i{\bf S}_{i,j}=1$ with some elements may be negative.

If $\Phi$ is also unital, that is, $\Phi(\mathbf{1}_\inp) = \mathbf{1}_\out$, 
then ${\bf s}$ is bistochastic and Eq.~(\ref{eq:S}) reduces to  ${\bf S}_{i,j}= (d_\inp +1){\bf s}_{i,j} - {d_\out}^{-1}$.
The examples of PTP maps for the qubit case ($d=2$) are presented in Appendix~\ref{app:PTP}.

Let now $\Phi$ be a CPTP map that is representing a quantum channel. 
Any such map gives rise to a Kraus representation
\begin{equation}\label{eq:krausopact}
	\rho^\out=\Phi[\rho^\inp] = \sum_k A_k\rho^\inp A_k^\dagger, \quad \sum_k A_k^\dagger A_k=\mathbf{1}_{\rm in} .
\end{equation}
In this case Eq.~(\ref{eq:S}) may be rewritten as follows:
\begin{equation}
	{\bf S}_{i,j} = (d_\inp +1){\bf s}_{i,j}-\frac{1}{d_\out} \tr\left[\Pi^\out_i \sum_k A_k A_k^\dagger\right]
\end{equation}
where
${\bf s}_{i,j} = d_\out^{-1}\sum_k\tr\left[A_k\Pi_i^\inp A_k^\dagger\Pi^\out_j\right]$.
Again, ${\bf s}$ is a stochastic matrix, while ${\bf S}$ is pseudostochastic.

Let us observe that we can define ${\bf S}$ in another way.
Let us rewrite Eq.~\eqref{eq:krausopact} in the vectorized form:
\begin{equation} \label{eq:Krauss}
	\kett{\rho_{\rm out}}=\sum_k A_k\otimes A_k^*\kett{\rho_{\rm in}}={\bf A}\kett{\rho_{\rm in}},
\end{equation}
where ${\bf A}=\sum_k A_k\otimes A_k^*$.
Consequently, we obtain ${\bf K}_\out p^\out = {\bf A} {\bf K}_\inp p^\inp$,
where ${\bf K}_\out$ and ${\bf K}_\inp$ are defined via Eq.~\eqref{eq:Kdef}, which is related to the corresponding SIC-POVM elements.
Finally, we have the following expression for the pseudostochastic matrix {\bf S}:
${\bf S} = {\bf K}_\out^{-1} {\bf A} {\bf K}_\inp$.

We can also define $\Phi$ using the Choi state
\begin{equation}
	\rho_\Phi = \frac{1}{d_\inp}\sum_{i,j}\ket{i}\bra{j}\otimes\Phi[\ket{i}\bra{j}].
\end{equation}
The vectorized version of the Choi state then reads $\mathbf{R} = \frac{1}{d_\inp}\sum_{i,j}\kett{\Phi[\ket{i}\bra{j}]}\bra{i}\otimes\bra{j}$.
In this form, the output state can be calculated in the following way: $\kett{\rho_{\rm out}} = d_\inp\mathbf{R} \kett{\rho_{\rm in}}$.
By comparing this result with Eq.~\eqref{eq:Krauss}, we obtain
\begin{equation} \label{eq:Choi_rec}
	{\bf R}=\frac{1}{d_\inp}{\bf A}_\Phi=
	\frac{1}{d_\inp}{\bf K}_{\rm out}\mathbf{S}{\bf K}_{\rm in}^{-1}.
\end{equation}
and $ {\bf S}=d_{\rm in}{\bf K}_{\rm out}^{-1}{\bf R}{\bf K}_{\rm in}$.

It is important to note that $\Phi$ is CPTP map if and only if $\rho_\Phi\geq 0$ and $\tr_{\mathcal{H}_{\rm in}}\rho_\Phi={\bf 1}/d_{\rm in}$.
These requirements allow one to check whether a given pseudostochastic matrix ${\bf S}$ corresponds to any CPTP map.

\section{Dynamics of SIC-POVM probability vectors}\label{sec:dynamics}

Here we study the dynamics of the SIC-POVM probability vector $p$.
For this purpose, 
we analyze (i) the representation for dynamical equation for $p$ corresponding to the unitary evolution of the density matrix governed by the von Neumann equation, 
(ii) the representation for unitary operators as the solution of dynamical equations,
and (iii) the representation for dissipative evolution for $p$ corresponding to Markovian master equation.
We place the summary of results of each following subsection in Table~\ref{tbl:1}.

\begin{table*}[t]
	\begin{tabular}{|p{0.23\linewidth}|p{0.3\linewidth}|p{0.47\linewidth}|}
		\hline
		& Standard representation & Probability representation for SIC-POVM $\{\Pi_i/d\}_{i=1}^{d^2}$\\ \hline
		$d$-dimensional quantum state &
		$\rho$ -- $d\times d$ Hermitian unit-trace semi-positive complex matrix &
		$p=\begin{bmatrix}
		p_1 & \ldots \ p_{d^2}
		\end{bmatrix}^{\rm T}$
		is the real probability vector with nonnegative elements\\
		&
		$\tr\rho=1$, $\rho\geq 0$. &
		$\sum_{i=1}^{d^2}p_i=1$, $p_i\geq 0$.\\
		\hline
		Measurement results probabilities for a POVM
		$\{E_i\}$
		&
		$q_i = \tr(\rho E_i)$
		&
		$q = {\bf M} p$, where
		${\bf M} = (d+1){\bf m}-{\bf E}$,	${\bf E}_{i,j}=\tr E_i$, \newline
		${\bf m}_{i,j}=\tr(E_i\Pi_j)$
		(${\bf M}$ is pseudostochastic matrix).	
		\\ \hline
		Evolution equation defined by a Hamiltonian $H$ &
		${\rm i}\dot{\rho}=[H,\rho]$ \newline ($H$ is Hermitian) &
		$\dot{p}={\bf H} p$, where
		$\mathbf{H}_{i,j} = (d+1)d^{-1}\tr(H[\Pi_j,\Pi_i])$ \newline
		(${\bf H}$ is real antisymmetric).
		\\ \hline
		Solution of the evolution equation &
		$\rho = U(t)\rho^{\rm in}U^\dagger(t)$, where\newline
		$U(t) = T\left\{\exp\left(-{\rm i}\int_{t'=0}^{t}H(t')dt'\right)\right\}$,	
		\newline
		($U(t)$ is unitary).
		&
		$p = {\bf U}(t)p^{\rm in}$, where \newline
		${\bf U}(t) = T\left\{\exp\left(\int_{t'=0}^{t}{\bf H}(t')dt'\right)\right\}$
		\newline
		(${\bf U}(t)$ is unitary and pseudobistochastic). \\
		\hline
		Markovian master equation governed by the GKSL generator
		&
		$\dot{\rho}=-{\rm i} [H,\rho]+$\newline $+ \sum_k \left( V_k \rho V_k^\dagger  - \frac 12 [V_k^\dagger V_k\rho +  \rho V_k^\dagger V_k] \right)$
		&
		$\dot{p}= \mathbf{L}p$, where \newline
		${\bf L}={\bf K}^{-1}\Lambda{\bf K}$,
		$\Lambda = -{\rm i}(C \otimes \mathbf{1} - \mathbf{1} \otimes C^*) + \sum_k V_k \otimes V_k^*$,
		$C = H - \frac{\rm i}{2} \sum_k V_k^\dagger V_k$.		
		\\
		\hline
		Quantum channel defined by Kraus operators $\{A_i\}$ &
		$\rho^{\rm out} =\sum_{i}A_i\rho^{\rm in} A_i^\dagger$
		&
		$p^{\rm out} = {\bf S} p^{\rm in}$, where
		${\bf S}= (d_{\rm in}+1){\bf s}- d_\out^{-1} \tr\left[\Pi^\out_i \sum_k A_k A_k^\dagger\right]$, \newline ${\bf s}_{i,j}=d_{\rm out}^{-1}\sum_k\tr\left(A_k\Pi_i^\inp A_k^\dagger\Pi^\out_j\right)$ \newline
		($\bf S$ is pseudostochastic). \\
		\hline
	\end{tabular}
	\caption{Correspondence between standard and SIC-POVM probability representations.}
	\label{tbl:1}
\end{table*}

\subsection{Stochastic representation of the von Neumann equation}\label{ssec:stochrep}

Consider a standard von Neumann evolution equation
\begin{equation}\label{eq:vNeq}
	{\rm i}\hbar\dot{\rho}=[H,\rho],
\end{equation}
where $[\cdot,\cdot]$ stands for commutator, $i^2=-1$, $\dot\rho$ is the time derivative of the state $\rho$, $H=H^\dagger$ is the system Hamiltonian, and $\hbar$ is the Plank constant.
In the general case, the Hamiltonian $H$ may depend on time $t$.
In what follows, we use dimensionless units and set $\hbar:=1$.

To make a transition from the density matrix $\rho$ to the corresponding SIC-POVM probability vector $p$, we multiply both sides of Eq.~\eqref{eq:vNeq} by $\Pi_i$ and take the trace.
Thus, we obtain
\begin{equation}\label{eq:vNeq2}
	{\rm i}\tr(\dot{\rho}\Pi_i)=\tr([H,\rho]\Pi_i).
\end{equation}
Taking into account Eq.~\eqref{eq:property_of_K}, the left-hand side of Eq.~\eqref{eq:vNeq2} can be rewritten as follows:
\begin{equation} \label{eq:LHS}
	{\rm i}\tr(\dot{\rho}\Pi_i)={\rm i}\sum_j\tr(K_j\Pi_i)\dot{p_j}=
	 i d\dot{p_i}.
\end{equation}
The right-hand side of Eq.~\eqref{eq:vNeq2} can be written in the following way:
\begin{multline} \label{eq:RHS}
	\tr\left([H,\rho]\Pi_i\right)=\sum_j\tr\left([H,((d+1)\Pi_j-\mathbf{1})p_j]\Pi_i\right)\\=\sum_jp_j(d+1)\tr([H,\Pi_j]\Pi_i)\\
	=\sum_jp_j(d+1)\tr(H[\Pi_j,\Pi_i]).
\end{multline}
By combining Eqs.~\eqref{eq:LHS} and ~\eqref{eq:RHS}, we obtain
\begin{equation}\label{eq:probeq}
\dot{p}={\bf H} p,
	\end{equation}
where ${\bf H}$ is a $d^2\times d^2$ matrix with elements given by
$\mathbf{H}_{i,j} = \frac{d+1}{{\rm i}d}\tr(H[\Pi_j,\Pi_i])$.
One can see that the elements of ${\bf H}$ can be also rewritten as
\begin{equation} \label{eq:HbviaH}
	\mathbf{H}_{i,j} = \frac{2(d+1)}{d} {\rm{Im}}\left[\braket{\psi_j}{\psi_i}\bra{\psi_i}H\ket{\psi_j}\right],
\end{equation}
where $\rm{Im}\left[\cdot\right]$ stands for the imaginary part.

Let us also consider an alternative approach to obtaining the form of ${\bf H}$.
We can write the von Neumann equation~\eqref{eq:vNeq} in the vectorized form:
${\rm i} \kett{\dot{\rho}} = (H\otimes {\bf 1}- {\bf 1}\otimes H^*) \kett{\rho}$.
Then by taking into account Eq.~\eqref{eq:statetoprob}, we have
${\rm i} {\bf K}\dot{p} = (H\otimes {\bf 1}- {\bf 1}\otimes H^*) {\bf K}{p}$.
Comparing this result with Eq.~\eqref{eq:probeq}, we obtain the  following representation of the Hamiltonian:
\begin{equation} \label{eq:hbanother}
	{\bf H}=-{\rm i} {\bf K}^{-1} (H\otimes {\bf 1}- {\bf 1}\otimes H^*) {\bf K}.
\end{equation}

The matrix ${\bf H}$ in the form given by Eq.~\eqref{eq:HbviaH} has a number of important properties.
\begin{enumerate}
	\item {\bf H} is real and antisymmetric: $\mathbf{H}_{i,j}=-\mathbf{H}_{ji}$.
	\item Diagonal elements and trace of {\bf H} are zero:
	$\mathbf{H}_{ii}=0$, $\tr\mathbf{H}=0$.	
	\item  Each row and column of {\bf H} summing to 0: $\sum_{i}\mathbf{H}_{i,j}=0$, $\sum_{j}\mathbf{H}_{i,j}=0$
	(here we employed the fact that $\sum_i\ket{\psi_i}\bra{\psi_i}={\bf 1}d$).
\end{enumerate}

The number of independent parameters defining the $d^2\times d^2$ matrix with such properties is equal to $N_{\bf H}=(d^2-1)(d^2-2)/2.$
Meanwhile, the physical properties of the Hamiltonian is defined with $N_{H}=d^2-1$ parameters (the term $-1$ comes from the fact that the Hamiltonian is defined up to a term proportional to the identity matrix).
One can see that $N_{\bf H}>N_{H}$ for $d>2$, so there should be some additional constraints on ${\bf H}$ on the top of the listed properties.

In order to study these constraints, we introduce a set $\{\sigma^{(j)}\}_{j=1}^{d^2-1}$ of orthogonal (with respect to the Hilbert-Schmidt distance) traceless Hermitian matrices in $\mathcal{L}(\mathcal{H})$ (e.g., Pauli matrices for $d=2$) and a renormalized identity matrix $\sigma^{(d^2)}:=\sqrt{2/d}{\bf 1}$ satisfying the relation $\tr(\sigma^{(i)}\sigma^{(j)})=2\delta_{i,j}$ for any $i,j\in\{1,\ldots,d^2\}$.

Then the Hamiltonian can be written in the form
\begin{equation} \label{eq:Hrep}
	H = \sum_{i=1}^{d^2} \lambda_i \sigma^{(i)},
\end{equation}
where $\lambda_i=\tr\left(H\sigma^{(i)}\right)/2$.
We note that all $\{\lambda_i\}_{i=0}^{d^2-1}$ can take arbitrary real values.

Substituting representation~\eqref{eq:Hrep} into Eq.~\eqref{eq:hbanother}, we obtain
\begin{eqnarray}
	{\bf H} &=& \sum_{i=0}^{d^2-1}\lambda_i\left(-{\rm i} {\bf K}^{-1} \left(\sigma^{(i)}\otimes {\bf 1}- {\bf 1}\otimes \sigma^{(i)*}\right) {\bf K}\right) \nonumber \\
	&=& \sum_{i=1}^{d^2-1}\lambda_i{\bf H}^{(i)},
\end{eqnarray}
where ${\bf H}^{(i)} = -{\rm i} {\bf K}^{-1} \left(\sigma^{(i)}\otimes {\bf 1}- {\bf 1}\otimes \sigma^{(i)*}\right) {\bf K}$.
We see that the parameter $\lambda_0$ does not participate in defining ${\bf H}$.

It turns out that the matrices from the set $\{{\bf H}^{(i)}\}_{i=1}^{d^2}$ are orthogonal to each other:
\begin{multline}
	\tr\left({\bf H}^{(i)} {\bf H}^{(j)T}\right)=-
	\tr\left({\bf H}^{(i)} {\bf H}^{(j)}\right)\\
	=\tr\left(\sigma^{(i)}\sigma^{(j)}\otimes {\bf 1} +
	{\bf 1} \otimes  \sigma^{(i)*}\sigma^{(j)*}\right)=4d\delta_{i,j}.	
\end{multline}
Thus, the set $\{{\bf H}^{(i)}\}$ forms a basis of the $(d^2-1)$-dimensional linear subspace of real antisymmetric matrices in the space of all antisymmetric matrices corresponding to physical processes.

We can also introduce a projector $\mathcal{P}_{\rm Ham}$ on this subspace, which acts on arbitrary real $d^2\times d^2$ matrix $\widetilde{\bf H}$ as follows:
\begin{equation} \label{eq:projector}
	\mathcal{P}_{\rm unit} (\widetilde{{\bf H}}) =  \frac{1}{4d} \sum_{i=1}^{d-1}\tr({\bf H}^{(i)T}\widetilde{{\bf H}}) {\bf H}^{(i)}.
\end{equation}
Therefore, a matrix ${\bf H}$ corresponding to a physical Hamiltonian satisfies the relation
\begin{equation}\label{eq:Hamprojcond}
	\mathcal{P}_{\rm unit}({\bf H})={\bf H}.
\end{equation}
This is an important relation for the representation of the SIC-POVM probability vector dynamics.

\subsection{Representation of unitary operators}\label{ssec:repun}

If the matrix ${\bf H}$ does not depend on $t$ (the Hamiltonian $H$ does not depend on time $t$), then the solution of the basic dynamics equation for the SIC-POVM probability vector Eq.~\eqref{eq:probeq} is as follows:
\begin{equation}
	p(t) = {\bf U}(t)p^{\rm in}, \quad {\bf U}(t)= \exp\left({{\bf H}t}\right),
\end{equation}
where $p^{\rm in}$ is a SIC-POVM probability vector at $t=0$.

At the same time, we know that the solution for $\rho$ via standard evolution operator is given by:
\begin{equation}
	\rho(t) = U(t)\rho^{\rm in}U^\dagger(t), \quad U(t)=\exp\left({-{i}Ht}\right),
\end{equation}
where $\rho_{\rm in}$ is initial density matrix (corresponded to $p^{\rm in}$).
Then we can obtain
\begin{equation}
	p_i(t)=\sum_j\left((d+1){\bf u}_{i,j}(t)-\frac{1}{d}\right)p_j^{\rm in},
\end{equation}
where
${\bf u}_{i,j}(t) = \frac{1}{d}\tr[U(t)\Pi_jU^\dagger(t)\Pi_i]$.
It easy to see that ${\bf u}$ is a $d^2 \times d^2$ bistochastic matrix:
\begin{equation}
	\sum_i {\bf u}_{i,j}(t)=\sum_j {\bf u}_{i,j}(t) =1, \quad {\bf u}_{i,j}(t)\geq 0.
\end{equation}

Finally, we arrive at the following expression:
\begin{equation} \label{eq:psm_for_unitary}
	{\bf U}(t) = (d+1){\bf u}(t)-\frac{1}{d}\widetilde{\bf I},
\end{equation}
where $\widetilde{\bf I}$ is the $d^2\times d^2$ matrix with all elements equal to unity (i.e., ${\bf I}_{i,j}=1$).

We can highlight the following basic properties of the operator ${\bf U}(t)$.
\begin{enumerate}
	\item ${\bf U}(t)$ is pseudobistochastic:
	\begin{equation}
	\sum_i {\bf U}_{i,j}(t)=\sum_j {\bf U}_{i,j}(t) =1,
	\end{equation}
	and some its elements of may be negative.

	\item ${\bf U}(t)$ is orthogonal:
	\begin{equation}
		{\bf U}(t){\bf U}^{\rm T}(t)={\bf U}^{\rm T}(t){\bf U}(t)={\mathbf I},
	\end{equation}
	where ${\mathbf I}$ is a $d^2\times d^2$ identity matrix.
\end{enumerate}
We note that according to these two properties an action of ${\bf U}(t)$ preserves both $l_1$ and $l_2$ norms.

Finally, the generalization on time-dependent Hamiltonians can be provided in a straightforward way.
In this case of $t>0$, the evolution operator is given by
\begin{equation}
	{\bf U}(t) = T\left\{\exp\left(\int_{t'=0}^{t}{\bf H}(t')dt'\right)\right\},
\end{equation}
where $T$ is the standard time-ordering operator.
It is easy to see that the mentioned properties for ${\bf U}(t)$ remain true in this case.

We present an example for the construction of ${\bf H}$ and ${\bf U}(t)$ in Appendix~\ref{app:HU}.

\subsection{Stochastic representation of dissipative evolution} \label{ssec:Linblad}

Consider now the Markovian master equation
\begin{equation}\label{eq:origME}
	\dot{\rho}= L(\rho),
\end{equation}
which is governed by the Gorini-Kossakowski-Sudarshan-Lindblad (GKSL) generator \cite{Gorini1976,Linblad1976} of the following form:
\begin{equation}\label{}
	L(\rho) = -{\rm i} [H,\rho] + \sum_k \left( V_k \rho V_k^\dagger  - \frac 12 [V_k^\dagger V_k\rho +  \rho V_k^\dagger V_k] \right),
\end{equation}
where $H$ is Hamiltonian and $V_k$ are so-called noise operators.
By fixing an orthonormal basis $|i\rangle$ in $\mathcal{H}$, one defines the following matrix
\begin{equation}\label{Kij}
	K_{i,j} := {\rm Tr}\left[ P_i L(P_j) \right] ,
\end{equation}
where $P_i = |i\rangle \langle i|$. The matrix $K_{i,j}$ defines a Kolmogorov generator~\cite{Kampen} as follows: $K_{i,j} \geq 0$, $i \neq j$,
and $\sum_i K_{i,j} = 0$, $j=1,\ldots,d$. 
Let us then define the stochastic representation of Eq.~\eqref{eq:origME} using SIC-POMV. 
Equation~(\ref{eq:origME}) gives rise to the following equation for the probability vector $p$:
\begin{equation}\label{eq:ME}
	\dot{p}= \mathbf{L}p ,
\end{equation}
where the $d^2\times d^2$ matrix $\mathbf{L}$ is defined as follows:
\begin{equation}\label{S}
	{\bf L}_{i,j} := (d +1)\frac 1d {\rm Tr}\left[ \Pi_i L(\Pi_j) \right] - \frac{1}{d} {\rm Tr}\left[ \Pi_i L(\mathbf{1}) \right] .
\end{equation}
The matrix ${\bf L}_{i,j}$ satisfies the condition
$\sum_{i=1}^{d^2} {\bf L}_{i,j} = 0$, $j=1,\ldots,d^2$; 
however, condition ${\bf L}_{i,j} \geq 0$ for $i\neq j$ needs not be satisfied.
We may call ${\bf L}_{i,j}$ a {\em pseudo-Kolmogorov generator}.  

Finally, defining the vectorization of $L$
\begin{equation}\label{}
	\Lambda := -{ i}(C \otimes \mathbf{1} - \mathbf{1} \otimes C^*) + \sum_k V_k \otimes V_k^* ,  
\end{equation}
with $C = H - \frac{\rm i}{2} \sum_k V_k^\dagger V_k$, 
one finds  the following representation for the pseudo-Kolmogorov generator:
${\bf L} = {\bf K}^{-1}  \Lambda {\bf K}$.

The solution of the master equation~\eqref{eq:ME} takes the form
\begin{equation}
	p(t) = {\bf S}(t)p^{\rm in},
\end{equation}
where the pseudostochastic matrix ${\bf S}$ is given by ${\bf S}(t) = \exp({\bf L}t)$
for time-independent generator ${\bf L}$ and
\begin{equation}
	{\bf S}(t) = T\left\{\exp\left(\int_{t'=0}^{t}{\bf L}(t')dt'\right)\right\},
\end{equation}
if there is time dependence in the generator.

Finally, we introduce a projector on the set of real matrices corresponded to time-independent Markovian noise generators in Eq.~\eqref{eq:ME}.
For this purpose, we represent the generator ${\bf L}$ in the form ${\bf L} = {\bf H} + {\bf D}$,
where the first term ${\bf H}$ is a Hamiltonian part, defined by Eq.~\eqref{eq:hbanother}, and ${\bf D}$ is a remaining part corresponding to the Markovian noise.
Next, we refer to ${\bf H}$ and ${\bf D}$ as unitary evolution generator and Markovian noise generators correspondingly.

Consider a decomposition of the noise operators in the form $V_k = \sum_{i=1}^{d^2} v_k^{(i)}$,
where $v_k^{(i)}$ are arbitrary complex numbers.
Then the operator ${\bf D}$ takes the form
\begin{equation}
	{\bf D} = \sum_{i,j} \Omega_{i,j} \sum_k v_k^{(i)}{v_k^{(j)*}},
\end{equation}
where
\begin{multline}
	\Omega_{i,j} = {\bf K}^{-1}\left(\sigma^{(i)}\otimes\sigma^{(j)*}- \frac{1}{2}\sigma^{(j)}\sigma^{(i)}\otimes{\bf 1}\right.\\
	\left.-\frac{1}{2}{\bf 1}\otimes\sigma^{(i)*}\sigma^{(j)*}\right){\bf K}.
\end{multline}
Noticing that $\sum_{k}v_k^{(i)}{v_k^{(j)*}}$ can be considered as $(i,j)$th element of some semipositive matrix ${\bf v}$, we can introduce a complex matrix ${\bf V}$ such that
\begin{equation}
	[{\bf V}{\bf V}^{\dagger}]_{i,j}={\bf v}_{i,j}=\sum_{k}v_k^{(i)}{v_k^{(j)*}}
\end{equation}
(hereafter, $[{\bf a}]_{i,j}$ stands for a $(i,j)$th matrix element of a matrix ${\bf a}$).
Although the dimensionality of ${\bf V}$ can be reduced down to $d^2\times {\rm rank}({\bf v})$, we prefer to treat ${\bf V}$ as $d^2\times d^2$ matrix with possibly zero columns.

Using the matrix ${\bf V}$ as a parametrization of the Markovian noise generator ${\bf D}$, we define a function
\begin{equation}
	{\bf D}({\bf V}):=\sum_{i,j}[{\bf V}{\bf V}^{\dagger}]_{i,j}\Omega_{i,j}.
\end{equation}
It allows us to introduce a function (projector) which outputs a matrix of physical time-independent Markovian generator which is closest, in terms of the Frobenius norm, to given $d^2\times d^2$ matrix $\widetilde{\bf D}$:
\begin{equation}
	\mathcal{P}_{\rm Mark}(\widetilde{\bf D}) :={\bf D}\left(\argmin_{{\bf V}\in\mathcal{M}_{d^2}}\tr\left[\left( {\bf D}({\bf V})-\widetilde{\bf D} \right)^2\right]\right),
\end{equation}
where $\mathcal{M}_{d^2}$ is set of complex $d^2\times d^2$ matrices.
We note, that it is practical to consider ${\bf V}$ in the form ${\bf V}={\bf V}_{\rm re}+{\rm i}{\bf V}_{\rm im}$, where ${\bf V}_{\rm re}$ and ${\bf V}_{\rm im}$ are real matrices, and perform an optimization over two spaces of real matrices.

So the necessary and sufficient condition for a generator ${\bf L}$ to correspond to time-independent Markovian evolution is to satisfy a relation
\begin{equation} \label{eq:general_relation_on_L}
	\mathcal{P}_{\rm unit}({\bf L})+\mathcal{P}_{\rm Mark}({\bf L}-\mathcal{P}_{\rm unit}({\bf L})) = {\bf L}
\end{equation}
which is a generalization of the condition~\eqref{eq:Hamprojcond}.

\section{Relation between pseudostochasticity and nonclassicality}
\label{sec:noncl}

Here we apply the developed SIC-POVM probability representation to study nonclassical features of quantum system dynamics.
We consider a general scheme of a quantum-mechanical experiment (see Fig.~\ref{fig:general_experiment}), which has three parts:
(i) preparation of a quantum state controlled by a random classical input $X$ taking values from some finite set $\mathcal{X}$;
(ii) evaluation of this stated under the sequence of quantum channels; and
(iii) measurement of the resulting state, which provides a (generally random) outcome $Y$ from some finite set $\mathcal{Y}$.
Let us consider these steps in more detail and in the framework of the pseudostochastic representation.
We assume that it is possible to construct SIC-POVMs for all involved dimensions of quantum systems.

Let the input $X$ have a probability distribution $p_X(x)$ ($x\in\mathcal{X}$).
We can also introduce a corresponding probability vector $p_X$ with $|\mathcal{X}|$ components given by values $p_X(x)$.
Let for every $x\in \mathcal{X}$ the prepared state in the case $X=x$ is described by a $d_{\rm in}\times d_{\rm in}$ density matrix  $\rho_{x}$. 
Then the preparation of the quantum state can be described with $d_{\rm in}^2\times |\mathcal{X}|$ stochastic matrix ${\bf G}$, each of whose columns is a SIC-POVM vector $p^{{\rm in},x}$ corresponded to the state $\rho_x$.
By multiplying ${\bf G}$ on the vector $p_X$ we obtain a SIC-POVM vector $p^{\rm in}=\sum_{x\in\mathcal{X}}p_X(x)p^{{\rm in},x}$, which corresponds to a prepared state $\sum_{x\in\mathcal{X}}p_X(x)\rho_x$.
So one can see that every quantum state preparation process can be considered as a classical probabilistic process.
However, the opposite statement is not true.
Each of columns of ${\bf G}$ strictly belongs to $d_{\rm in}^2$-dimensional qplex, so the set of possible stochastic matrices corresponded to quantum state preparation is strictly inside the set of all possible $d_{\rm in}^2 \times |\mathcal{X}|$ stochastic matrices.

We then consider an evolution of the prepared state during actions of quantum channels.
As discussed above, the action of each channel can be described by a multiplication on a pseudostochastic matrix on an input SIC-POVM probability vector.
Therefore, the total action of the channels on the quantum states is given by a product of pseudostochastic matrices, which is also pseudostochastic.
Let the output state be described with the $d_{\rm out}\times d_{\rm out}$ density matrix.
Then the action of quantum channels is given by the corresponding $d_{\rm out}^2 \times d_{\rm in}^2$ pseudostochastic matrix ${\bf S}$ and the resulting state is given by a SIC-POVM probability vector $p^{\rm out}={\bf S}p^{\rm in}$. 

Finally, we consider a measurement process whose output is described by a discrete random variable $Y$ taking values from a finite set $\mathcal{Y}$.
Let ${\bf M}$ be a pseudostochastic matrix of the measurement under consideration.
Then  $\mathcal{Y}$-dimensional probability vector $p_Y$ of the random variable $Y$ can be obtained as  $p_Y={\bf M}p^{\rm in}={\bf Q}p_X$, where ${\bf Q}={\bf M}{\bf S}{\bf G}$ is a $|\mathcal{Y}|\times|\mathcal{X}|$ stochastic matrix.
It is easy to see that by varying ${\bf M}$, ${\bf S}$, and ${\bf G}$ it is possible to construct any desired stochastic matrix ${\bf Q}$.

From this analysis, one can note that the non-classicality of the considered process is related to negative conditional probabilities of pseudostochastic maps, which correspond to quantum channels and quantum measurement.
Otherwise, if all matrices, which are involved in our consideration, are stochastic, then we can conclude that the whole process appears to be classical-like stochastic process. 
In this case, obtaining of measurement results in the real quantum experiment can be simulated by a classical sampling technique:
One can first sample a classical random variable $x$ from the distribution $p_X$, then sample a random variable from the distribution $p^{{\rm in},x}$, next 
sample a random variable from the certain column of the stochastic matrix corresponding to the first channel, and so on, until the final measurement.
We note that this technique turns out to be invalid in the case of pseudostochastic matrices with negative elements due to the necessity of sampling from distributions with negative elements.
Below we consider formulate necessary and sufficient for obtaining non-classical behavior of the resulting process.
We would like to note an alternative consideration of non-classical features of quantum system dynamics in the context of SIC-POVM is presented in Ref.~\cite{DeBrota2018}.

\subsection{Characterizing nonclassicality of time-independent Markovian dynamics}

Here we consider a time-independent Markovian dynamics governed by Eq.~\eqref{eq:ME}, where ${\bf L}$ is time independent and satisfies relation~\eqref{eq:general_relation_on_L}.
We raise the question about whether the resulting evolution demonstrates nonclassical features related to negative elements in the pseudostochastic matrix ${\bf S}(t)=e^{{\bf L}t}$ or it is classical-like.
The results are as follows.

\begin{theorem} \label{thm:1}
	Consider a time-independent Markovian master equation governed by the generator ${\bf L}$ given in the SIC-POVM probability representation.
	The resulting map ${\bf S}(t)=e^{{\bf L}t}$ is stochastic for any $t>0$ if and only if all nondiagonal elements of ${\bf L}$ are non-negative: ${\bf L}_{i,j}\geq 0$ for any $i\neq j$.
\end{theorem}
\begin{proof}
	First, we prove that if all nondiagonal elements of ${\bf L}$ are non-negative then ${\bf S}(t)$ is stochastic.
	In order to do so, we represent ${\bf S}(t)$ for some fixed $t$ in the form
	\begin{equation} \label{eq:series1}
		e^{{\bf L}t} = \lim_{n\rightarrow+\infty} \left({\bf I}+{\bf L}\frac{t}{n}\right)^n.
	\end{equation}
	Let $\widetilde{l}:=\max_i|{\bf L}_{i,i}|$.
	Since ${\bf L}_{i,j}\geq 0$ for $i\neq j$, all elements of the matrix ${\bf I}+{\bf L}\frac{t}{n}$ are non-negative for $n\geq \lceil \widetilde{l}t\rceil$.
	So the expansion~\eqref{eq:series1} is product of matrices, whose elements are non-negative, and therefore the resulting matrix also consists of non-negative elements only.
	
	Then, we prove that because ${\bf S}(t)$ is stochastic for all $t>0$, then all non-diagonal elements of ${\bf L}$ are non-negative.
	We prove this statement by contraposition.
	Let ${\bf L}_{i^*,j^*}<0$ for some $i^*,j^*$.
	Consider small values of $t$, and an expansion of ${\bf S}(t)$ in the form
	\begin{equation}
		{\bf S}(t) = {\bf 1}+{\bf L}t+\Delta{\bf L}(t), \quad \Delta{\bf L}(t)_{i,j} \in O(t^2).
	\end{equation}
	Since $\Delta{\bf L}(t)_{i^*j^*}$ is of the second order of smallness, it is possible to find small enough $t>0$ such that ${\bf L}_{i^*j^*}t+\Delta{\bf L}(t)_{i^*j*}<0$ and thus ${\bf S}_{i^*j^*}(t)<0$.
\end{proof}

The above theorem has an important corollary regarding to noiseless unitary processes with ${\bf L}=\mathcal{P}_{\rm unit}({\bf L})={\bf H}$.

\begin{corollary}
	Any non-zero unitary evolution generator ${\bf H}\neq0$, given in the SIC-POVM probability representation, spawns a pseudostochastic matrix ${\bf S}(t)=e^{{\bf H}t}$ with at least one negative element for some $t>0$.
\end{corollary}
\begin{proof}
	The statement directly follows from facts that ${\bf H}_{i,j}=-{\bf H}_{ji}$ and ${\bf H}$ is non-zero matrix.
\end{proof}

Thus, as one may expect, the necessary condition for quantum evolution to become classical-like is the presence of decoherence processes, which appear, e.g., due to Markovian noise.
We also note that for some values of $t>0$ the resulting evolution matrix ${\bf S}(t)=e^{{\bf H}t}$ can be stochastic; 
e.g., in the case of two-level system of Rabi oscillations ${\bf S}(t)$ is stochastic and equals to the identity matrix for $t=nT$, where $n$ is a positive integer and $T$ is a period of oscillations.
However, between these discrete moments ${\bf S}(t)$ appears to be pseudostochastic with negative elements.

The result of Theorem~\ref{thm:1} allows us to construct a measure of nonclassicality of time-independent Markovian dynamics determined by a generator ${\bf L}$. 
Let 
\begin{equation}
	\mathcal{N}(\widetilde{{\bf L}})=\max_{i\neq j}\left(\max(0,-\widetilde{{\bf L}}_{i,j}) \right)
\end{equation} 
be a magnitude of the smallest negative non-diagonal element of some matrix $\widetilde{{\bf L}}$.
In line with the results of Theorem~\ref{thm:1}, $\mathcal{N}({\bf L})$ seems to be suitable for the characterization of the nonclassicality of ${\bf L}$.
However, it is not invariant under change of the basis of the underlying $d$-dimensional Hilbert space determined by applying sone unitary operator $U$ to computational basis vectors.
It is easy to see that a transformation $\ket{\psi_i} \rightarrow U\ket{\psi_i}$ of basic SIC-POVM vectors $\{\ket{\psi_i}\}_{i=1}^{d^2}$ yields a transformation ${\bf L}\rightarrow {\bf U}{\bf L} {\bf U}^{\rm T}$, 
where ${\bf U}$ is a pseudostochastic representation of $U$.
Consequently, we can introduce the measure of nonclassicality $\delta_{\rm quant}({\bf L})$ for given generator ${\bf L}$  in the following form:
\begin{equation}\label{eq:nonclas}
	\delta_{\rm quant}({\bf L}) := \max_{{\bf U}\in\mathcal{U}} \mathcal{N}({\bf U}{{\bf L}}{\bf U}^{\rm T}),
\end{equation} 
where the set $\mathcal{U}$ of all possible pseudostochastic matrices corresponding to unitary operators is given by
\begin{equation}
	\mathcal{U}=\{\exp{{\bf H}}:{\bf H}=\sum_{i=1}^{d^2}\lambda_i{\bf H}^{(i)}, \lambda_i\in\mathbb{R}\}.
\end{equation} 

\begin{figure}
	\includegraphics[width=1\linewidth]{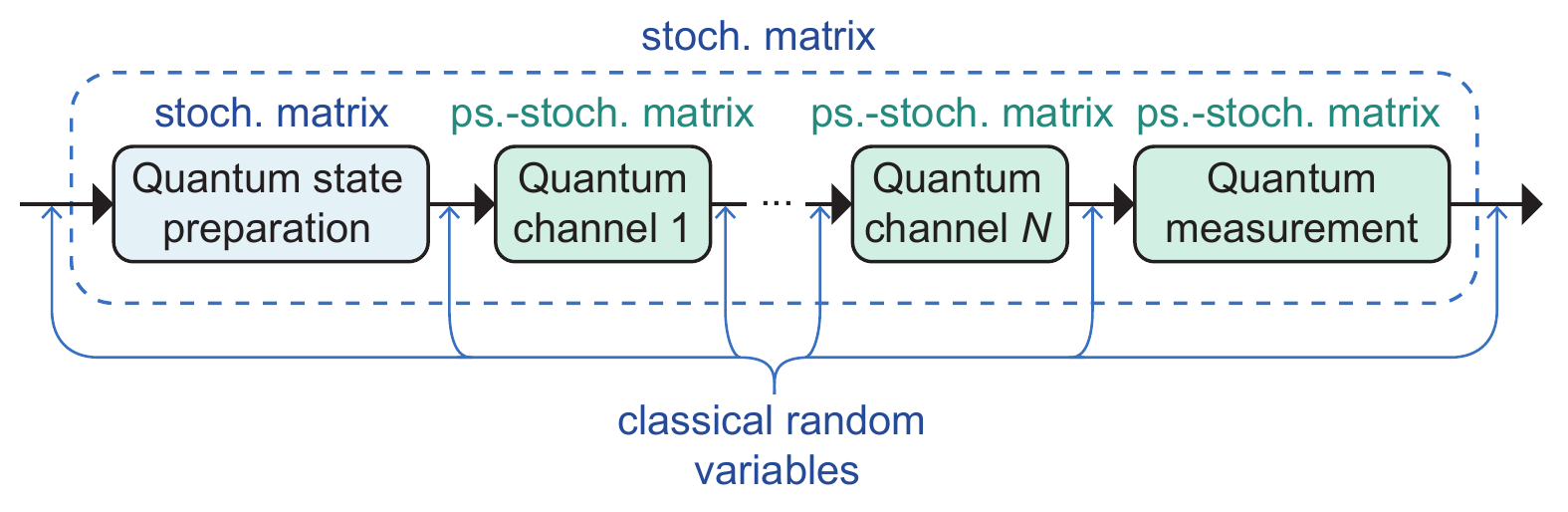}
	\vskip -4mm
	\caption{Scheme of a quantum-mechanical experiment.}
	\label{fig:general_experiment}
\end{figure}

\subsection{Measuring non-Markovianity of the dynamics}

In the above consideration, we stressed on a particular case of time-independent Markovian generator and discussed its nonclassicality.
Here we introduce an additional quantity, which indicates how accurately the considered Markovian type of dynamics can describe the observed dynamics, which is given by a pseudostochastic matrix ${\bf S}$.
In other words, we are considering a measure non-Markovianity of quantum processes.
We note that our approach in general is similar to the one proposed in Ref.~\cite{Wolf2008}.

Consider a projection of ${\bf S}$ on a set of pseudostochastic matrices corresponding to Markovian evolution with time-independent generators:
\begin{equation}
	{\bf S}_{\rm Mark}:=\exp\left[{\mathcal{P}_{\rm unit}({\bf L})+\mathcal{P}_{\rm Mark}({\bf L}-\mathcal{P}_{\rm unit}({\bf L}))}\right],
\end{equation}
where ${\bf L}:=\log{{\bf S}}$.
We can define the non-Markovianity measure $\delta_{\rm nMark}({\bf S})$ as
\begin{equation} \label{eq:nonmark}
	\delta_{\rm nMark}({\bf S}) := \frac{1}{d^2}\sqrt{\mathcal{D}({\bf S},{\bf S}_{\rm Mark})},
\end{equation}
where $\mathcal{D}({\bf S}_{1},{\bf S}_{2})= {\rm Tr}\left[({\bf S}_{1}-{\bf S}_{2})^2\right]$ is the Hilbert-Schmidt distance between matrices ${\bf S}_1$ and ${\bf S}_2$.
The values of $\delta_{\rm nMark}({\bf S})$ have a clear physical meaning: They show an average difference between elements of ${\bf S}$ and corresponding Markovian projection ${\bf S}_{\rm Mark}$.
We note that all elements ${\bf S}_{i,j}$ can be measured experimentally with some  absolute accuracy $\delta$.
So in order to conclude with confidence that ${\bf S}$ exhibit non-Markovianity, one should have $\delta_{\rm nMark}({\bf S})$ larger than $\delta$.

\section{Experimental study of superconducting circuits}
\label{sec:experiment}

In this section, we apply the SIC-POVM probability representation to experimental study of quantum processes realizing during implementation of quantum gates of the IBM QX4 cloud-based superconducting quantum processor~\cite{IBM}.
We consider a qubit dynamics determined by implementation of single-qubit gate, followed by a SIC-POVM measurement.
The SIC-POVM measurement is realized using an additional qubit in the pure state
\begin{equation}
	\ket{\phi}\bra{\phi}= \frac{1}{2}
	\begin{bmatrix}
		1+1/\sqrt{3} & 1/\sqrt{3} - {\rm i}/\sqrt{3} \\
		1/\sqrt{3} + {\rm i}/\sqrt{3} & 1-1/\sqrt{3}
	\end{bmatrix},
\end{equation}
applying controlled-NOT (CNOT) gate followed by the Hadamard gate, and performing standard projective measurements in the computational basis on both qubits [see Fig.~\ref{fig:two_qubit_circuit}(a)].
We note that alternative ways of performing SIC-POVM measurements can also be used~\cite{Oszmaniec2019}.
Let $ab$ be a classical two-bit string produced at each launch of the considered measurement scheme.
It is easy to check that there is one-to-one correspondence between possible outcomes $ab=00, 01, 10, 11$ and SIC-POVM effects $\Pi_1/2$, $\Pi_2/2$, $\Pi_3/2$, and $\Pi_4/2$ defined by projectors
\begin{equation} \label{eq:qubitSICPOVM}
	\begin{aligned}
		& \Pi_1 =\ket{\psi_1}\bra{\psi_1} = \frac{1}{2}\left[\sigma^{(4)}+\frac{1}{\sqrt{3}}(\sigma^{(1)}-\sigma^{(2)}+\sigma^{(3)})\right], \\
		& \Pi_2 =\ket{\psi_2}\bra{\psi_2} = \frac{1}{2}\left[\sigma^{(4)}+\frac{1}{\sqrt{3}}(\sigma^{(1)}+\sigma^{(2)}-\sigma^{(3)})\right], \\
		& \Pi_3 = \ket{\psi_3}\bra{\psi_3} = \frac{1}{2}\left[\sigma^{(4)}+\frac{1}{\sqrt{3}}(-\sigma^{(1)}+\sigma^{(2)}+\sigma^{(3)})\right], \\
		& \Pi_4 = \ket{\psi_4}\bra{\psi_4} = \frac{1}{2}\left[\sigma^{(4)}+\frac{1}{\sqrt{3}}(-\sigma^{(1)}-\sigma^{(2)}-\sigma^{(3)})\right]
	\end{aligned}
\end{equation}
(here $\sigma^{(1)}$, $\sigma^{(2)}$ and $\sigma^{(3)}$ stand for standard $x$-, $y$- and $z$- Pauli matrices respectively and $\sigma^{(4)}$ is a $2\times2$ identity matrix).
Bloch vectors of SIC-POVM projectors and the state $\ket{\phi}\bra{\phi}$ are shown in Fig.~\ref{fig:Bloch_sphere_new}.

\begin{figure}[t]
	\includegraphics[height=0.5\linewidth]{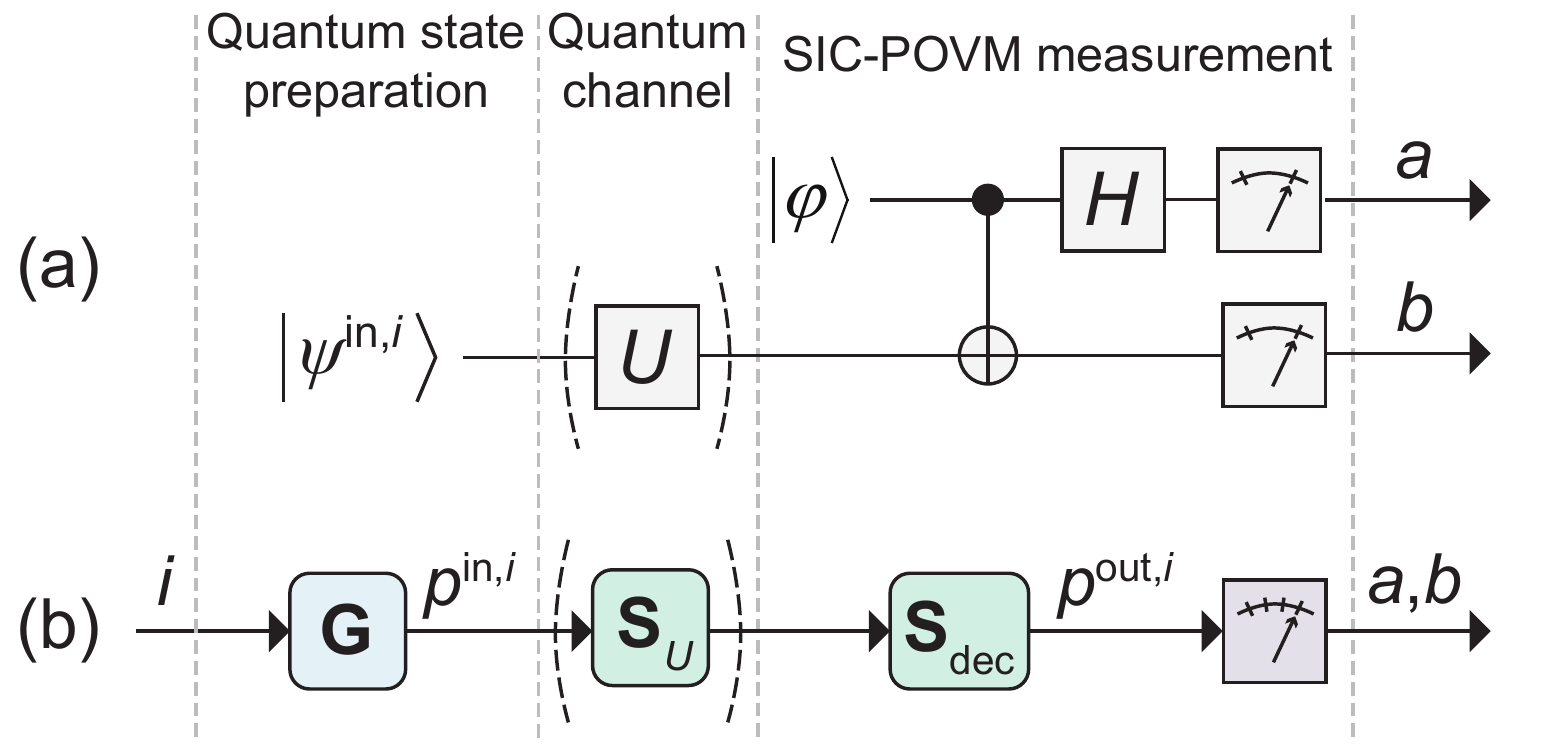}
	\caption{The experimental setup.
	For the tomography of the SIC-POVM measurement operation an applying of $U$ (${\bf S}_U$) is skipped.
	In panel (a), standard representation as a circuit is presented (standard notations for CNOT and Hadamard gates are used).
	In panel (b), SIC-POVM probability representation is shown.}
	\label{fig:two_qubit_circuit}
\end{figure}

\begin{figure}[h]
	\includegraphics[height=0.5\linewidth]{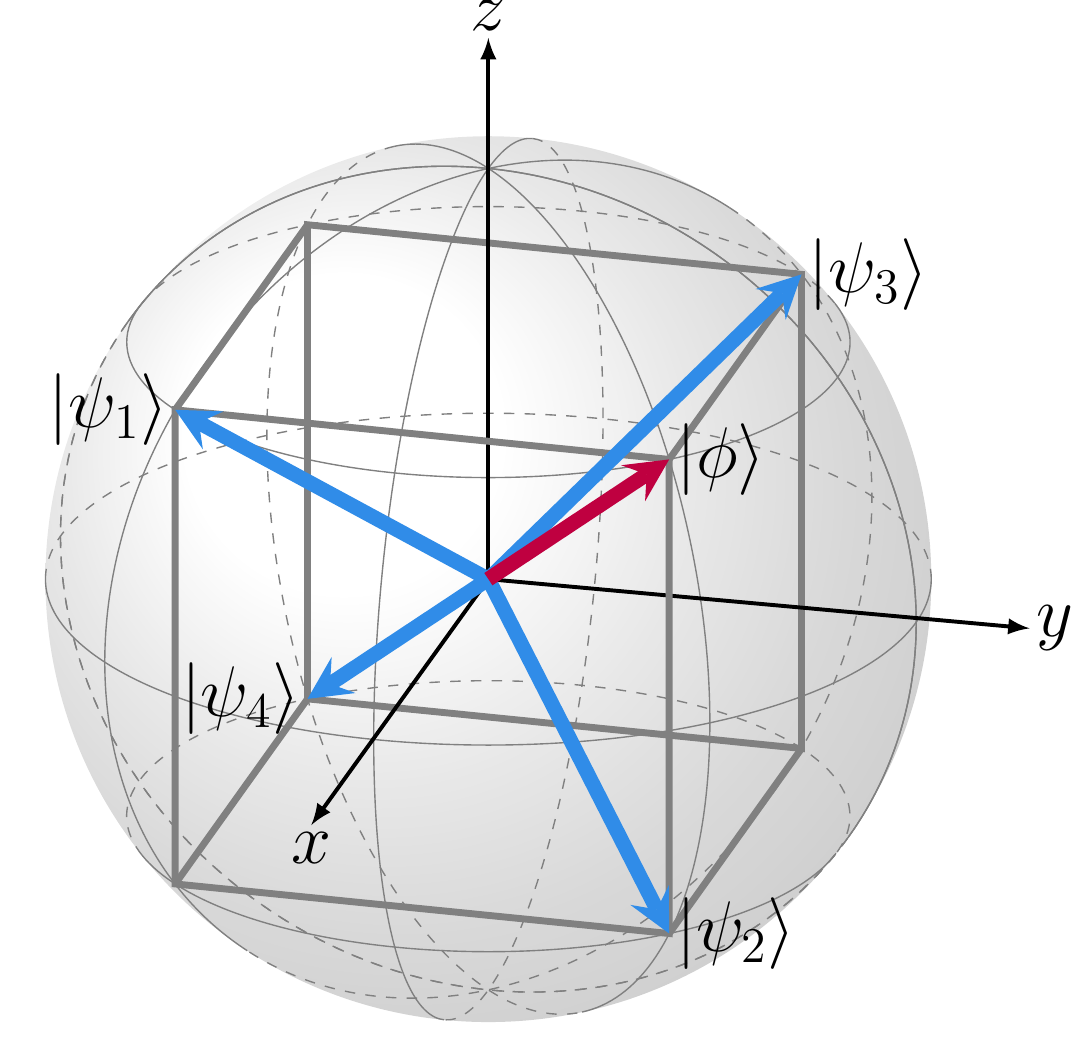}
	\caption{The Bloch vectors corresponded to SIC-POVM projectors $\{\ket{\psi_i}\bra{\ket{\psi_i}}\}_{i=1}^4$ and the state $\ket{\phi}\bra{\phi}$.}
	\label{fig:Bloch_sphere_new}
\end{figure}

Our experiment is designed as the tomographic reconstruction of the quantum channel related to the implementation of a single-qubit gate of the following form:
\begin{equation} \label{eq:U}
	U=\begin{bmatrix}
	1 & 0\\0 & {\rm i}
	\end{bmatrix},
\end{equation}
also known as the $S$ gate.
In order to perform the tomography protocol, 
we apply $U$ to a set of four input states $\{\ket{\psi_i}\}_{i=1}^4$--the same which were used for SIC-POVM construction--and then measure resulting states with the designed SIC-POVM measurement.

The corresponding scheme of the experiment in the SIC-POVM probability representation is presented in Fig.~\ref{fig:two_qubit_circuit}(b).
The quantum state preparation is presented by stochastic matrix 
${\bf G}=\begin{bmatrix}p^{\inp,1} & p^{\inp,2} & p^{\inp,3} & p^{\inp,4}
\end{bmatrix}$, 
where each $p^{\inp,i}$ is a SIC-POVM probability vector, corresponding to the state $\ket{\psi_i}\bra{\psi_i}$, with elements 
\begin{equation} \label{eq:pin_components}
	p^{\inp,i}_j={\rm Tr}(\Pi_i\Pi_j)/2=(2\delta_{i,j}+1)/6.
\end{equation}
We note that for a particular known value of $i$ the probability distribution is given by $p^{\inp,i}$.
Then this probability distribution is multiplied by a pseudostochastic matrix ${\bf S}_U$ corresponded to the gate $U$.
In the case of ideal realization, it would be given by
\begin{equation}
	{\bf S}^{\rm ideal}_U =  \frac{1}{2}
	\begin{bmatrix}
		1 & -1 & 1 & 1 \\
		1 & 1  & -1 & 1 \\
		1 & 1 & 1 & -1 \\
		-1 & 1 & 1 & 1 \\
	\end{bmatrix}.
\end{equation}
Because of experimental imperfections and decoherence, ${\bf S}_U$ has some forms different from ${\bf S}^{\rm ideal}_U$, and, actually, the purpose of the experiment is to find out.
We note that ${\bf S}^{\rm ideal}_U$ possesses negative elements of maximal possible for this dimensionality magnitude equal to $1/d=1/2$ [see Eq.~\eqref{eq:psm_for_unitary}].

Then the multiplication by a pseudostochastic matrix, which corresponds to the SIC-POVM measurement comes.
In the ideal case, it would be described by the $4\times 4$ identity matrix. 
However, the realistic experimental conditions it can also suffer from different kinds of imperfections, especially because of employing two-qubit CNOT gate.
In order to take imperfection into account, 
we model the real SIC-POVM measurement with a sequence of a decoherence channel described by some pseudostochastic matrix ${\bf S}_{\rm dec}$ and ideal SIC-POVM measurement described with the identity matrix.
The measurement outcome $ab$ is obtained by sampling a random variable from the resulting probability distribution $p^{{\rm out},i}={\bf S}_{\rm dec}\,{\bf S}_Up^{{\rm in},i}$.

For each input state, which is given by $i\in\{1,2,3,4\}$, we ran the experiment $N=1024$ times and calculate the output four-dimensional probability vectors $p^{\out,i}$ as frequencies of obtaining corresponding outcomes in SIC-POVM measurement.
Then, by employing known elements of input probability vectors $p^{\inp,i}_j$,
we obtain a complete system of 12 linear equations on the elements of a pseudostochastic matrix  ${\bf S}_{\rm dec}{\bf S}_U$ from general matrix equations
\begin{equation}
	p^{\out,i}=({\bf S}_{\rm dec}\,{\bf S}_U)~p^{\inp,i}, \quad i=1,2,3,4.
\end{equation}
We note that the $d^2\times d^2$ (pseudo)stochastic matrix is characterized by $d^2(d^2-1)$ independent elements due to the normalization requirement for each column.

Let ${\bf S}^{\rm raw}_{{\rm dec}, U}$ be a solution of the the system of linear equations.
The reconstruction error of each element of pseudostochastic matrices can be estimated at a level of $\delta\approx 0.031$ as shown in Appendix~\ref{app:errors}.
However, the resulting matrix is not guaranteed to be physical.
It means that the corresponding Choi matrix may have negative eigenvalues.
In order to obtain a physical result with certainty we introduce an operation of projection on a set of pseudostochastic maps corresponding to CPTP maps.
We derive it in a way similar to the derivation of $\mathcal{P}_{\rm Mark}(\cdot)$.

We consider the Kraus representation of an arbitrary CPTP map presented in Eq.~(\ref{eq:krausopact}).
Taking into account that each Kraus operator can be decomposed in the form $A_k = \sum_i a_k^{(i)} \sigma^{(i)}$ with complex coefficient $a_k^{(i)}$, we obtain a corresponding pseudostochastic matrix in the form
${\bf S} = \sum_{i,j} \sum_k a^{(i)}_k a^{(j)*}_k \Theta_{i,j}$, where 
$\Theta_{i,j}= {\bf K}_{\rm out}^{-1} \sigma^{(i)}\otimes\sigma^{(j)} {\bf K}_{\rm in}$.
Then we can treat $\sum_k a^{(i)}_k a^{(j)*}_k$ as elements of some semipositive matrix ${\bf V}{\bf V}^\dagger$.
This trick allows us to introduce a function
\begin{equation}
	{\bf S}({\bf V})=\sum_{i,j} [{\bf V}{\bf V}^{\dagger}]_{i,j} \Theta_{i,j}
\end{equation}
which parametrizes a CPTP map with complex $d^2 \times d^2$ matrix ${\bf V}$.
Finally, we introduce a projection operator
\begin{equation}
	\mathcal{P}_{\rm CPTP}(\widetilde{\bf S}) ={\bf S}\left(\argmin_{{\bf V}\in\mathcal{M}_{d^2}}\tr\left[\left( {\bf S}({\bf V})-\widetilde{\bf S} \right)^2\right]\right),
\end{equation}
which gives a physical pseudostochastic matrix closes to some matrix ${\widetilde{\bf S}}$.

We employ the constructed projector for obtaining ${\bf S}_{{\rm dec}, U} := \mathcal{P}_{\rm CPTP}({\bf S}^{\rm raw}_{{\rm dec}, U})$.
Then in order to get ${\bf S}_{\rm dec}$ and ${\bf S}_U$ separately, we performed the same experiment without implementation of the gate $U$.
It allows reconstructing ${\bf S}_{\rm dec}$ and obtaining ${\bf S}_U$ with the use of the inverse matrix ${\bf S}_{\rm dec}^{-1}$ and known matrix ${\bf S}_{{\rm dec}, U}$.
We note that the similar trick is used for improving results of quantum state and process tomography~\cite{Bantysh2018,Maciejewski2019}.
Because of the same number of circuit runs $N=1024$ the experimental error of the reconstruction $\delta$ remained the same.

The results of ${\bf S}_{\rm dec}$ and ${\bf S}_U$ reconstruction are presented in the first row of Table~\ref{tab:exp_rslts} and Fig.~\ref{fig:exp_rslts}.
The values of fidelities regard to the identity matrix and $U$ [see \eqref{eq:U}] equal 0.89 and 0.94, respectively.

\begin{figure}
	\includegraphics[width=\linewidth]{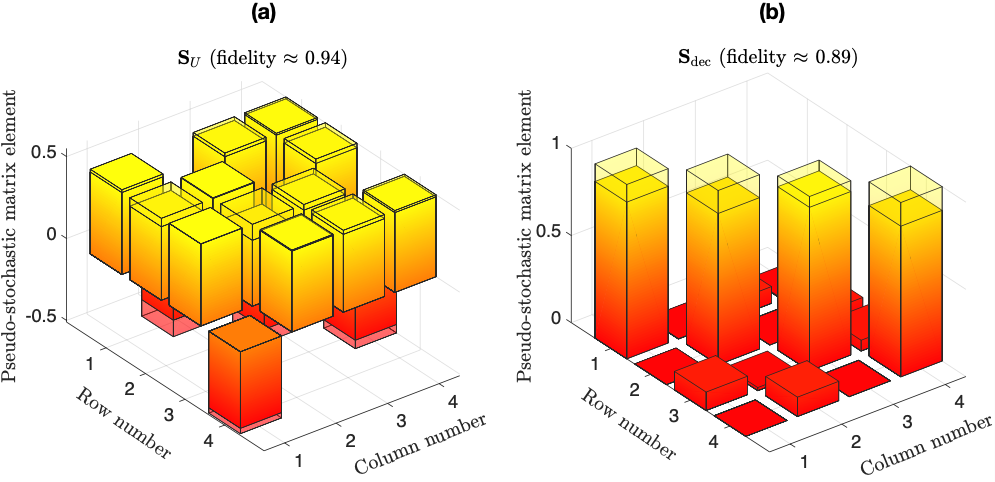}
	\caption{Experimentally reconstructed pseudostochastic matrices of the single-qubit gate $U$ (a) and the imperfect SIC-POVM measurement (b).}
	\label{fig:exp_rslts}
\end{figure}

\begin{table*}
	\centering
	\begin{tabular}{|p{0.24\linewidth}|p{0.38\linewidth}|p{0.38\linewidth}|} \hline
		& Gate implementation & SIC-POVM measurement \\ \hline
		Pseudostochastic matrix &
		${\bf S}_{U}=\begin{bmatrix}
			0.517 & -0.399 &  0.478 &  0.489 \\
			0.449 & 0.504  & -0.418 &  0.473 \\
			0.502 & 0.403  &  0.466 & -0.448 \\ 
			-0.467 & 0.493  &  0.475 &  0.487 
		\end{bmatrix}$ &
		${\bf S}_{\rm dec}=\begin{bmatrix}
			0.893 & 0.002 & 0.100 & 0.010 \\
			0.018 &  0.877 & 0.012 & 0.071 \\
			0.102 &  0.018 & 0.924 & 0.058\\
			0.014 & 0.111 & -0.002 & 0.874	
		\end{bmatrix}$ \\ \hline
		Unitary evolution generator & 
		${\bf H}_{U}=\begin{bmatrix}
		0     & -0.768 &-0.033  &  0.801 \\
		0.768 &  0     &-0.761  & -0.007 \\
		0.033 &  0.761 & 0      &-0.794\\
		-0.801 &  0.007 & 0.794  & 0   
		\end{bmatrix}$ &
		${\bf H}_{\rm dec}=\begin{bmatrix}
		0     &  -0.010 &  0.012 & -0.002 \\
		0.010  &  0      &  0.011 & -0.021 \\
		-0.012 &-0.011   &  0    & 0.023 \\
		0.002 & 0.021   & -0.023 & 0    
		\end{bmatrix}$		
		\\ \hline	
		Markovian evolution generator &
		${\bf D}_{U}=\begin{bmatrix}
		-0.044 &  0.046 &  0.070 &  0.001 \\
		0.013 & -0.105 &  0.005 &  0.030 \\
		0.033 &  0.001 & -0.110 &  0.003 \\
		-0.004 &  0.058 &  0.035 & -0.034 
		\end{bmatrix}$ &
		${\bf D}_{\rm dec} = \begin{bmatrix}
		-0.131 & 0.008 & 0.089 &  0.007\\
		0.007 &-0.136 &-0.001 & 0.098 \\
		0.113 & 0.025 &-0.102 & 0.035 \\
		0.011 & 0.103 & 0.014 & -0.139
		\end{bmatrix}$
		\\ \hline
		Non-Markovianity & 
		$\delta_{\rm nMark}({\bf S}_U)=0.003$ &
		$\delta_{\rm nMark}({\bf S}_{\rm dec})=0.007$ \\ \hline
		Non-classicality &
		$\delta_{\rm quant}({\bf H}_U+{\bf D}_U)=0.781$ &
		$\delta_{\rm quant}({\bf H}_{\rm dec}+{\bf D}_{\rm dec})=0$ \\ \hline
	\end{tabular}
	\caption{Experimental results.}
	\label{tab:exp_rslts}
\end{table*}

We then study a possibility to describe the observed processes with a time-independent Markovian master equation.
For this purpose, we compute $\ln({\bf S}_{\rm dec})$ and $\ln({\bf S}_U)$ and then extract unitary evolution and Markovian evolution parts by employing corresponding projectors.
We note that the unitary term is extracted first (e.g., ${\bf H}_U=\mathcal{P}_{\rm unit}[\ln({\bf S}_{\rm dec})]$), and then $\mathcal{P}_{\rm Mark}$ is applied to the orgthogonal part (e.g. ${\bf D}_{U}=\mathcal{P}_{\rm Mark}[\ln({\bf S}_{\rm dec})-{\bf H}_U]$).
We also computed measures of non-Markovianity given by $\delta_{\rm nMark}({\bf S}_{\rm dec})$ and $\delta_{\rm nMark}({\bf S}_U)$ [see Eq.~\eqref{eq:nonmark}] and measures of non-classicality for Markovian approximations, 
which are given by $\delta_{\rm quant}({\bf H}_U+{\bf D}_U)$ and $\delta_{\rm quant}({\bf H}_{\rm dec}+{\bf D}_{\rm dec})$ [see Eq.~\eqref{eq:nonclas}].

The obtained results are presented in Table~\ref{tab:exp_rslts}.
First of all, we note that the level of non-Markovianity for both processes is quite small compared to measurement error $\delta$.
Thus, the observed processes can be efficiently described with time-independent Markovian equations.
Second, we see that the imperfect SIC-POVM measurement can be described by purely classical stochastic process.
In the chosen basis, it has small nondiagonal negative elements; however, they can be completely removed by shifting to a basis: This fact corresponds to $\delta_{\rm quant}({\bf H}_{\rm dec}+{\bf D}_{\rm dec})=0$.
In contrast, the process of the gate implementation shows a clear nonclassical behavior: There are negative elements in the pseudostochastic matrix which can not be removed by any basis change.
So we see that SIC-POVM probability representation opens interesting possibilities for deep studying of quantum dynamics and revealing its classical and nonclassical features.

\section{Conclusion}\label{sec:conclusion}

In this work, we have considered a probability representation of quantum mechanics based on employing SIC-POVM.
In this representation, $d$-dimensional quantum systems are described with true $d^2$-dimensional probability distributions (probability vectors), while their dynamics and calculating measurement outcomes are governed by pseudostochastic matrices--matrices of conditional probabilities without restriction on elements positivity.

We have derived an equation of a SIC-POVM probability vector evolution corresponding to the von Neumann equation and the GKSL master equation.
In the case of unitary evolution, 
we have shown that the evolution generator in the SIC-POVM probability representation is given by $d^2\times d^2$ real antisymmetric matrix from a certain $(d^2-1)$ linear subspace of all $d^2\times d^2$ real antisymmetric matrices.
This fact has allowed us to construct a projector on the space of physical unitary evolution generators.
Then we have also shown that a corresponding evolution operator is the SIC-POVM probability representation is given by orthogonal pseudobistochastic matrix preserving both $l_1$ and $l_2$ norms.
In the case of the dissipative evolution, we have also shown that is possible to construct a projector on a set of time-independent Markovian dissipators in the SIC-POVM probability representation.

We have applied our results to studying non-classical features of quantum system dynamics.
We have proven the theorem about necessary and sufficient conditions on the evolution generator  which make the resulting evolution described with stochastic (but not pseudostochastic) matrices.
Then we have constructed a practical measures of nonclassicality and non-Markovianity for the observed quantum processes described with pseudostochastic matrices.

Finally, we have applied our approaches to the experimental study of superconducting quantum circuits run on the IBM quantum processor.
We have demonstrated that a noise appearing in the imperfect SIC-POVM measurement demonstrates classical behavior, while quantum dynamics during application of the single-qubit gate has clear nonclassical features.
Meanwhile, we have shown that both processes are well described by the Markovian approximation that is revealed by our non-Markovianity measure.

\section*{Acknowledgments}
	
We acknowledge use of the IBM Q Experience for this work.
The views expressed are those of the authors and do not reflect the official policy or position of IBM or the IBM Q Experience team.
The authors thank A. Ulanov, M. Gavreev, E. Tiunov, and D. Kurlov for fruitful discussions.
Results of Secs.~\ref{sec:general},~\ref{ssec:stochrep},~\ref{ssec:repun}, and ~\ref{sec:noncl} were obtained by E.O.K. and V.I.M with the support from the Russian Science Foundation Grant No. 19-71-10091.
Results of Sec.~\ref{ssec:Linblad} have been obtained by D.C. supported by Narodowe Centrum Nauki under the Grant No. 2018/30/A/ST2/00837.
Results of Sec.~\ref{sec:experiment} were obtained by A.K.F., A.O.M., and A.S.M. with the support from the Grant of the President of the Russian Federation (Project No. MK923.2019.2).
	
\appendix

\section{Relation between SIC-POVM probabilities and MUB probabilities for $d=2$}\label{app:MUB}

Here we provide a correspondence between the considered SIC-POVM probability representation and the alternative probability representation based projective MUB measurements in the case of two-level systems,
which were extensively studied in Refs.~\cite{Chernega2017,Chernega2018, Avanesov2019}.
In this representation, an arbitrary quantum state $\rho$ is characterized with a three-dimensional vector
\begin{equation}
\widetilde{p}:=
\begin{bmatrix}
\widetilde{p}_1\\ \widetilde{p}_2\\ \widetilde{p}_3,
\end{bmatrix},
\end{equation}
where
\begin{multline}
\begin{aligned}
\widetilde{p}_1 := \bra{+}\rho\ket{+},
\quad
\widetilde{p}_2 := \bra{R}\rho\ket{R},
\quad
\widetilde{p}_3 := \bra{0}\rho\ket{0}\\
\ket{+}:=\frac{1}{\sqrt{2}}(\ket{0}+\ket{1}),
\quad
\ket{R}:=\frac{1}{\sqrt{2}}(\ket{0}+{\rm i}\ket{1}).
\end{aligned}
\end{multline}
One can see that the probabilities $\{\widetilde{p}_i\}_{i=1}^{3}$ define projections on $x$, $y$, and $z$ axes of the Bloch sphere.
We note that all three components of $\widetilde{p}$ are independent.

In order to obtain the relation between the vector $\widetilde{p}$ and the SIC-POVM vector $p$ we can employ Eqs.~\eqref{eq:arbitrmeas} and ~\eqref{eq:Mdef}.
Finally, we obtain
\begin{equation}
	\widetilde{p}=
	{\bf F}
	p,
\end{equation}
where
\begin{equation}
	{\bf F}=
	\frac{1}{2}
	\left[\begin{array}{llll}{1+\sqrt{3}} & {1+\sqrt{3}} & {1-\sqrt{3}} & {1-\sqrt{3}} \\ {1-\sqrt{3}} & {1+\sqrt{3}} & {1+\sqrt{3}} & {1-\sqrt{3}} \\ {1+\sqrt{3}} & {1-\sqrt{3}} & {1+\sqrt{3}} & {1-\sqrt{3}}\end{array}\right],
\end{equation}
where we use the qubit SIC-POVM based on the projectors set~\eqref{eq:qubitSICPOVM}.

In order to obtain the opposite relation, we note that
$p_4 = 1-p_1-p_2-p_3$.
Then one can verify the following equality:
\begin{equation}\label{eq:forwardrel}
\begin{bmatrix}
\widetilde{p}_1\\\widetilde{p}_2\\\widetilde{p}_3\\1
\end{bmatrix}=
\frac{1}{2}\left[\begin{array}{cccc}{2 \sqrt{3}} & {2 \sqrt{3}} & {0} & {1-\sqrt{3}} \\ {0} & {2 \sqrt{3}} & {2 \sqrt{3}} & {1-\sqrt{3}} \\ {2 \sqrt{3}} & {0} & {2 \sqrt{3}} & {1-\sqrt{3}} \\ {0} & {0} & {0} & {2}\end{array}\right]
\begin{bmatrix}
p_1\\p_2\\p_3\\1
\end{bmatrix}.
\end{equation}

Using~\eqref{eq:forwardrel}, we obtain the following relation:
\begin{equation}
\begin{bmatrix}
p_1\\p_2\\p_3\\p_4
\end{bmatrix}
=
\frac{\sqrt{3}}{12}
\left[\begin{array}{cccc}{2} & {-2} & {2} & {\sqrt{3}-1} \\ {2} & {2} & {-2} & {\sqrt{3}-1} \\ {-2} & {2} & {2} & {\sqrt{3}-1} \\ {-2} & {-2} & {-2} & {\sqrt{3}(1+\sqrt{3})}\end{array}\right]
\begin{bmatrix}
\widetilde{p}_1\\\widetilde{p}_2\\\widetilde{p}_3\\1
\end{bmatrix}.
\end{equation}
It can be rewritten in a more compact form:
\begin{equation}
p = {\bf T}\widetilde{p}+c,
\end{equation}
where
\begin{equation}
{\bf T}=\frac{\sqrt{3}}{6}\begin{bmatrix}
1 & -1 & 1\\
1 & 1 & -1\\
-1 & 1 & 1\\
-1 & -1 &-1\\
\end{bmatrix},
\quad
c=\frac{1}{12}\begin{bmatrix}
3-\sqrt{3}\\3-\sqrt{3}\\3-\sqrt{3}\\3+3\sqrt{3}
\end{bmatrix}.
\end{equation}

\section{Examples of pseudostochastic matrices of PTP {\bf D}s for $d=2$} \label{app:PTP}

Here we consider two well-known maps in entanglement theory: transposition map $T(X) = X^{\rm T}$, and so-called reduction map $R : \mathcal{L}(\mathcal{H}) \to \mathcal{L}(\mathcal{H})$,
\begin{equation}
	R(X) = \frac{1}{d-1} (\mathbf{1} {\rm Tr}X - X), 
\end{equation}
which is unital PTP, but not CPTP. 
In the case of $d_\inp=d_\out=d=2$ for SIC-POVM effects given by~\eqref{eq:qubitSICPOVM}, one finds the pseudobistochastic matrices corresponding to transposition
\begin{equation}
	{\bf S}_T = \frac 12 \left[ \begin{array}{rrrr} 1 & 1 & 1 & -1 \\ 1 & 1 & -1 & 1 \\1 & -1 & 1 & 1 \\ -1 & 1 & 1 & 1 \end{array} \right],
\end{equation}
and to reduction map
\begin{equation}
	{\bf S}_R = \frac 12 \left[ \begin{array}{rrrr} -1 & 1 & 1 & 1 \\ 1 & -1 & 1 & 1 \\1 & 1 & -1 & 1 \\1 & 1 & 1 & -1 \end{array} \right].
\end{equation}

\section{Construction of pseudostochastic matrix of a unitary evolution for $d=2$} \label{app:HU}

Here we present explicit forms of some basic matrices used for SIC-POVM probability representation in the case $(d=2)$-dimensional systems (qubits) and provide an example of a unitary pseudobistochastic evolution operator construction.
We consider the SIC-POVM based on the projectors given in Eq.~\eqref{eq:qubitSICPOVM}.
The matrix ${\bf K}$ and its inverse, which defines transitions between SIC-POVM probability representation and the standard representation, take the following forms:
\begin{equation}
	{\bf K} = \frac{1}{2}
	\begin{bmatrix}
		1+\sqrt{3} & \quad 1-\sqrt{3} & \quad 1+\sqrt{3}  &\quad 1-\sqrt{3} \\
		\sqrt{3}+{i} \sqrt{3} & \sqrt{3}-{i} \sqrt{3} & -\sqrt{3}-{i} \sqrt{3} &-\sqrt{3}+{i} \sqrt{3}\\
		\sqrt{3}-{i}\sqrt{3} & \sqrt{3}+{i} \sqrt{3} & -\sqrt{3}+{i} \sqrt{3} & -\sqrt{3}-{} \sqrt{3} \\
		1-\sqrt{3} & 1+\sqrt{3} & 1-\sqrt{3} & 1+\sqrt{3}
	\end{bmatrix},
\end{equation}
\begin{multline}
	{\bf K}^{-1} = \\
	\frac{1}{12}
	\begin{bmatrix}
		1+\sqrt{3} & 1-\sqrt{3} & 1+\sqrt{3} & 1-\sqrt{3} \\
		\sqrt{3}+{i} \sqrt{3} & \sqrt{3}-{i} \sqrt{3} & -\sqrt{3}-{ i} \sqrt{3} & -\sqrt{3}+{i} \sqrt{3}\\
		\sqrt{3}-{i} \sqrt{3} & \sqrt{3}+{i} \sqrt{3} & -\sqrt{3}+{i} \sqrt{3} & -\sqrt{3}-{i} \sqrt{3}\\
		1-\sqrt{3} & 1+\sqrt{3} & 1-\sqrt{3} & 1+\sqrt{3}
	\end{bmatrix}.
\end{multline}

The antisymmetric matrices forming a basis for physical process take the forms
\begin{equation}
\begin{aligned}
& {\bf H}^{(1)}=
\left[ \begin{array}{rrrr}{0} & {0} & {1} & {-1} \\ {0} & {0} & {-1} & {1} \\ {-1} & {1} & {0} & {0} \\ {1} & {-1} & {0} & {0}\end{array}\right], \\
& {\bf H}^{(2)}=
\left[ \begin{array}{rrrr}{0} & {-1} & {1} & {0} \\ {1} & {0} & {0} & {-1} \\ {-1} & {0} & {0} & {1} \\ {0} & {1} & {-1} & {0}\end{array}\right], \\
&{\bf H}^{(3)}=
\left[ \begin{array}{rrrr}{0} & {-1} & {0} & {1} \\ {1} & {0} & {-1} & {0} \\ {0} & {1} & {0} & {-1} \\ {-1} & {0} & {1} & {0}\end{array}\right].
\end{aligned}
\end{equation}

Consider a Hamiltonian $H={\sigma_3}/{2}$.
In the probability representations the dynamics is defined by the matrix ${\bf H} = \frac{1}{2}{\bf H}^{(3)}$.
Then the resulting pseudobistochastic operator of the evolution reads
\begin{equation}
\begin{split}
\!\!\!\!\!\!\!\!&{\bf U}(t)=\exp\left({\bf H}t\right)\\
\!\!\!\!\!\!&=\frac{1}{2}
\left[ \begin{array}{cccc}{1+\cos (t)} & {-\sin (t)} & {1-\cos (t)} & {\sin (t)} \\ {\sin (t)} & {1+\cos (t)} & {-\sin (t)} & {1-\cos (t)} \\ {1-\cos (t)} & {\sin (t)} & {1+\cos (t)} & {-\sin (t)} \\ {-\sin (t)} & {1-\cos (t)} & {\sin (t)} & {1+\cos (t)}\end{array}\right],
\end{split}
\end{equation}
which corresponds to the standard evolution operator
\begin{equation}
	U(t)=\exp\left(-{ i}Ht\right)=\begin{bmatrix}
		\exp({ i}t/2) & 0\\
		0 & \exp(-{ i}t/2)
	\end{bmatrix}.
\end{equation}

\section{Estimating statistical errors in the experimental reconstruction of a pseudostochastic matrix} \label{app:errors}

In order to reconstruct an unknown pseudostochastic matrix ${\bf S}$ in our quantum process tomography protocol, we considered a set of equations
\begin{equation}
p^{{\rm out},i}_j=\sum_{k=1}^4 {\bf S}_{jk}p^{{\rm in},i}_k, \quad i,j=1,2,3,4
\end{equation}
where the values of $p^{{\rm in},i}_k$ are given by Eq.~\eqref{eq:pin_components}, and the values of $p^{{\rm out},i}_j$ are obtained from the experiment as frequencies of getting a $j$th SIC-POVM measurement outcome for $i$th input state.

In order to write down a set of linear equations on the components of ${\bf S}$ in a standard form, we can represent ${\bf S}$ as follows:
\begin{equation}
{\bf S}=\begin{bmatrix}
s^{(1)} & s^{(2)} & s^{(3)} & s^{(4)}
\end{bmatrix}^{\rm T},
\end{equation}
where $s^{(i)}$ are four-dimensional column vectors.
We note that due to normalization condition on the columns of ${\bf S}$ we have $s^{(4)}_j=1-s^{(1)}_j-s^{(2)}_j-s^{(3)}_j$, so we need to reconstruct $s^{(1)}$, $s^{(2)}$ and $s^{(3)}$ only.

Then by introducing
\begin{equation}
q^{(j)} := \begin{bmatrix}
p^{{\rm out},1}_j & p^{{\rm out},2}_j & p^{{\rm out},2}_j & p^{{\rm out},4}_j
\end{bmatrix}^{\rm T}
\end{equation}
and
\begin{equation}
{\bf P}:=\begin{bmatrix}
p^{{\rm in},1~{\rm T}}\\
p^{{\rm in},2~{\rm T}} \\
p^{{\rm in},3~{\rm T}} \\
p^{{\rm in},4~{\rm T}} \\	
\end{bmatrix}=
\frac{1}{6}\left[\begin{array}{llll}{3} & {1} & {1} & {1} \\ {1} & {3} & {1} & {1} \\ {1} & {1} & {3} & {1} \\ {1} & {1} & {1} & {3}\end{array}\right],
\end{equation}
we can write the basic system of linear equations on the components of ${\bf S}$ in a compact form
\begin{equation}
q^{(j)} = {\bf P}s^{(j)},\quad j=1,2,3.
\end{equation}
Its solution is then given by
\begin{equation}
s^{(j)}={\bf P}^{-1}q^{(j)}, \quad j=1,2,3,
\end{equation}
where
\begin{equation}
{\bf P}^{-1}=\frac{1}{2}
\left[\begin{array}{rrrr}{5} & {-1} & {-1} & {-1} \\ {-1} & {5} & {-1} & {-1} \\ {-1} & {-1} & {5} & {-1} \\ {-1} & {-1} & {-1} & {5}\end{array}\right].
\end{equation}

In order to estimate statistical errors of the obtained result, we first estimate statistical errors of reconstructed output probabilities.
The mean squared error for each of $q^{(i)}_j=p^{{\rm out},i}_j$ can be obtained in the following way:
\begin{equation}
\left(\delta q^{(i)}_j\right)^2\approx\frac{(1-p^{{\rm out},j}_i)p^{{\rm out},j}_i}{N}\leq \frac{1}{4N},
\end{equation}
where $N=1024$ is number measurements performed for each input state.
Then the resulting statistical error for components of ${\bf S}$ can be estimated as
\begin{equation}
	\delta s^{(j)}_i = \sqrt{ \sum_{k=1}^4 | {\bf P}^{-1}_{ik} | \left(\delta q^{(j)}_k\right)^2  }\leq\frac{1}{\sqrt{N}}=\frac{1}{32}\approx 0.031.
\end{equation}

\end{document}